%% file: arxiv.tex
\documentclass[11pt]{article}
\usepackage{hyperref} 
\usepackage{appendix}
\usepackage{amssymb}
\usepackage{amsfonts}
\usepackage{latexsym}
\usepackage{amsmath}
\usepackage{amsthm}
\usepackage{amsfonts}
\usepackage{yfonts}
\usepackage[T1]{fontenc}
\usepackage{fullpage,appendix}
\usepackage{algorithmic}
\usepackage[ruled,vlined]{algorithm2e}
\usepackage{dsfont}
\usepackage[capitalize]{cleveref}

\usepackage{color}
\usepackage{tikz}

\newcommand{\eqdef}{\stackrel{\textrm{def}}{=}}

\newcommand{\ignore}[1]{}
\definecolor{corlinks}{RGB}{64,128,128}
\definecolor{cormenu}{RGB}{0,37,94}
\definecolor{corurl}{RGB}{0,46,91}

\newcommand{\on}{\{-1,1\}}
\newcommand{\1}{\mathds{1}}

\newcommand{\E} {\mathbb{E}}

\newcommand{\R}{\mathbb{R}}

\newcommand{\N}{\mathbb{N}}

\newcommand{\C}{\mathcal{C}}

\newcommand{\girth}{\text{\frakfamily g}}

\newcommand{\I}{\mathcal{I}}

\newcommand{\F}{\mathbb{F}}

\newcommand{\zo}{\{0, 1\}}

\newcommand{\eps}{\epsilon}

\newtheorem{fact}{Fact}[section]

\newtheorem{theorem}[fact]{Theorem}
\newtheorem{lemma}[fact]{Lemma}
\newtheorem{corollary}[fact]{Corollary}
\newtheorem{observation}[fact]{Observation}

\newtheorem{claim}[fact]{Claim}

\newtheorem{remark}[fact]{Remark}
\theoremstyle{definition}
\newtheorem{definition}[fact]{Definition}

\newcommand{\bnote}[1]{\textcolor{red}{\small{}}}
\newcommand{\pnote}[1]{\textcolor{blue}{\small{}}}
\newcommand{\snote}[1]{\textcolor{orange}{\small{}}}

\newcommand{\tE}{\tilde{\mathbb{E}}}
\newcommand{\Span}{\mathsf{Span}}
\newcommand{\tchi}{\tilde{\chi}}

\newcommand{\dist}{\mathsf{dist}}
\newcommand{\cl}{\mathsf{cl}}

\newcommand{\bchi}{\bar{\chi}}

\newcommand{\e}{\epsilon}
\newcommand{\pmo}{\{\pm 1\}}
\newcommand\Hastad{H{\aa}stad}
\newcommand\cR{\mathcal{R}}
\newcommand\cD{\mathcal{D}}
\newcommand\cP{\mathcal{P}}
\newcommand\pE{\Tilde{\mathbb{E}}}

\newcommand{\stoc}[2]{#2}
\newcommand\stoceq[1]{$$#1$$}

\title{Sum of Squares Lower Bounds from Pairwise Independence}

\author{Boaz Barak\thanks{Microsoft Research New England} \and Siu On Chan\thanks{Microsoft Research New England} \and Pravesh Kothari\thanks{University of Texas, Austin. Work done while an intern at Microsoft Research New England.}}

\begin{document}
\thispagestyle{empty}
\setcounter{page}{0}
\maketitle

\input{abstract}

\clearpage

\input{introduction}

\input{overview}

\input{prelims}

\input{closure}

\input{independence}

\input{proof}
\section*{Acknowledgements}
Thanks to Ryan O'Donnell, Li-Yang Tan, and David Steurer for fruitful discussions and the anonymous reviewers for their valuable comments and suggestions on a previous version of this paper.

\bibliographystyle{plain}
\bibliography{refs}
\appendix
\input{appendix}

\end{document}

%% file: abstract.tex
\begin{abstract}
We prove that for every $\epsilon>0$ and predicate $P:\{0,1\}^k\rightarrow \{0,1\}$ that supports a pairwise independent distribution, there exists an
instance $\I$ of the \textsc{Max}$P$ constraint satisfaction problem on $n$ variables such that no assignment can satisfy more than a $\tfrac{|P^{-1}(1)|}{2^k}+\epsilon$
fraction of $\I$'s constraints but the degree $\Omega(n)$ Sum of Squares semidefinite programming hierarchy cannot certify that $\I$ is unsatisfiable. Similar results were previously only known for weaker hierarchies.
\end{abstract}

%% file: introduction.tex
\section{Introduction}

\emph{Constraint Satisfaction Problems} (CSP) are among the most natural computational problems, and yet their computational complexity is not fully understood.
In particular several works have studied the notion of \emph{Approximation Resistance}, which loosely speaking means that the best polynomial-time approximation
algorithm is simply the one that outputs a random assignment. Under Khot's \emph{Unique Games Conjecture}~\cite{Khot02} much is known about this property.
In particular Austrin and Mossell~\cite{AustrinM09} showed if the UGC is true, then, for every predicate $P:\zo^k\rightarrow \zo$, if there exists a \emph{pairwise independent} distribution $\mu$ over $P^{-1}(1)$ (i.e., a distribution $\mu$
such that for every $i\neq j \in [k]$, the marginal $\mu_i\mu_j$ is the uniform distribution over $\zo^2$), then $P$ is approximation resistant.
Austrin and \Hastad~\cite{AustrinH11} used this to establish (under the UGC) fairly tight bounds on the threshold at which a random predicate of a particular density becomes approximation resistant.
However, there is no consensus whether the UGC is true. Assuming only  $\mathbf{P}\neq\mathbf{NP}$, the best known bound is by Chan~\cite{Chan13} who showed that a predicate is approximation resistant if it contains a
distribution $\mu$ as above satisfying the additional condition that it is uniform over a subspace $V \subseteq GF(2)^k$. This algebraic structure is a fairly strong condition. In particular if we choose $P:\zo^k\rightarrow\zo$ to be a random predicate conditioned on  $|P^{-1}(1)|=t$ (where $t \in \{1\ldots 2^k\}$ is some parameter), then $P$ will satisfy the first condition (supporting a pairwise independent distribution)  with high probability as long as $t>ck^2$ for some constant $c$~\cite{AustrinH11} while it will \emph{not} satisfy the second condition even for $t$ as large as $\exp(k/5)$ (see Observation~\ref{obs:no-aff-plane}).

Another line of work has been concerned with proving \emph{unconditional} lower bounds for these problems on restricted families of algorithms. These works considered convex relaxations for CSPs, where we say that a CSP is \emph{approximation resistant} for some relaxation $\cR$ if there is an instance for which a random assignment is essentially optimal, but the relaxation value is $1-o(1)$ (namely, the relaxation ``thinks'' that it's possible to satisfy almost all constraints).  Interestingly, the unconditional results match the conditional ones. That is, for certain weaker relaxations (namely, the Sherali-Adams linear programming hierarchy or Sherali-Adams augmented with the basic semidefinite program), there are unconditional results for the same predicates that were shown approximation-resistant under the UGC~\cite{BGMT12,TulsianiW12,ODonnellW14}. (This is of course not a coincidence, as the UGC is intimately connected with some of these weaker relaxations~\cite{Raghavendra08}.) In contrast, for the stronger \emph{Sum of Squares (SOS)} (also known as \emph{Lasserre}) relaxation~\cite{Shor87,Nesterov00,Parrilo00,Lasserre01}, the previously known results~\cite{Grigoriev01,Schoenebeck08,Tulsiani09}  utilized the same conditions as in Chan's NP-hardness result (and in fact inspired Chan's work).

In this work we show that the pairwise independence condition suffices for lower bounds even for this stronger Sum-of-Squares hierarchy. This result is interesting in its own right and, based on past experience, could also be viewed as suggesting that it may be possible to improve the UGC-based results to results based on $\mathbf{P}\neq\mathbf{NP}$.

\subsection{Our results}

Our results actually hold for a more general setting than showing approximation-resistance of predicates, and so to state them we need to introduce some notation.
Roughly speaking, we show that for every $k$ and an arbitrarily small $\epsilon > 0$, there exists a set $\I = \{ C_1,\ldots,C_m \}$ of $k$-tuples of literals (i.e. variables or their negations) over the variables $x_1,\ldots,x_n$
such that \textbf{(1)} for every assignment $x$ to the variables, the induced distribution on $\zo^k$ obtained by taking a random $i\in [m]$ and looking at the literals in $C_i$ is $\e$-close to the uniform distribution on $\zo^k$ but
\textbf{(2)} for every pairwise independent distribution $\mu$ over $\zo^k$,  there is a relaxation-solution that ``cheats'' the $\Omega(n)$-degree SOS relaxation to think that there is a distribution $\cD$ over assignments (i.e.~$\zo^k$) such that for every $i\in [m]$, the projection of $\cD$ to the literals in $C_i$ is distributed according to $\mu$.
This immediately implies that predicates supporting a pairwise independent distribution are approximation-resistant for this relaxation.  We now formally state our results:

\begin{definition}[Pseudo-expectation]  For every $n$ and $d$, let $\cP_d^n$ denote the linear space of $n$-variate real polynomials of degree at most $d$.
A linear operator $\pE:\cP_d^n\rightarrow \R$ is a \emph{degree-$d$ pseudo-expectation operator} if it satisfies:

\begin{description}
\item[Normalization] $\pE[1]= 1$ where on the LHS $1$ denotes the constant polynomial $p$ such that $p(x)=1$.

\item[Positivity] $\pE [p^2] \geq 0$ for every $p\in\cP_{d/2}^n$.
\end{description}

For every polynomial $p\in \cP_d^n$, we say that $\pE$ satisfies the constraint $\{ p = 0 \}$ if $\pE [pq] = 0$ for every $q\in \cP^n_{d-\deg(P)}$.
\end{definition}

The Sum-of-Squares hierarchy can be thought of as optimizing over pseudo-expectations; see the survey~\cite{BarakS14} and the references therein, as well as the lecture notes~\cite{Barak14}.
For notational convenience, we will use variables over $\pmo$ instead of $\zo$. A \emph{literal} is a function $f:\pmo^n\rightarrow \pmo$ such that $f(x)=x_i$ or $f(x)=-x_i$ for some $i$.
If $C= (f_1,\ldots,f_k)$ is a $k$-tuple of literals then we denote by $C(x)$ the tuple $(f_1(x),\ldots,f_k(x))$. Our main result is the following:

\begin{theorem}[Main Result] \label{thm:main} For every $k\in\N$, $\e>0$
there exists $\delta = \delta(k) >0$ such that for every  sufficiently large $n\in\N$ there is a set  $\I = \{ C_1,\ldots,C_m \}$ of $k$-tuples of literals over $x_1,\ldots,x_n$  such that

\begin{enumerate}

\item For every $x\in\pmo^n$, the distribution $\{ C(x) \}$ where $C$ is chosen at random in $\I$ is within $\e$ statistical distance to the uniform distribution over $\pmo^k$.

\item For every pairwise independent distribution $\mu$ over $\pmo^k$, there exists a degree $\delta n$ pseudo-expectation operator $\pE$ over $\R^n$ satisfying the constraints $\{ x_j^2 = 1 \}_{j=1\ldots n}$ such that for every $C\in \I$ and $f:\pmo^k\rightarrow\R$, $\pE f(C(x)) = \E f(\mu)$.
\end{enumerate}
\end{theorem}

The following immediate corollary implies that predicates supporting pairwise independent distributions are approximation-resistant for $\Omega(n)$-degree SOS:

\begin{corollary} For every $\e>0$ and $P:\pmo^k\rightarrow \zo$, if there exists a pairwise independent distribution $\mu$ supported on $P^{-1}(1)$ then there exists $\delta>0$ such that for all $n$ there is a set  $\I = \{ C_1,\ldots,C_m \}$ of $k$-tuples of literals over $x_1,\ldots,x_n$  such that
 \begin{enumerate}

\item For every $x\in\pmo^n$,  $\E_{C\in \I} P(C(x)) \leq \tfrac{|P^{-1}(1)|}{2^k} + \e$.

\item The value of the $\delta n$-degree Max-$P$ SOS relaxation for the fraction of satisfiable constraints on the instance $\I$ is $1$.
 \end{enumerate}

\end{corollary}

\begin{remark} The instance $\I = (C_1,\ldots,C_m)$ is actually obtained at random (with some pruning of a small fraction of the constraints, or alternatively, with some loss in the
``perfect completeness'' condition).  Thus our results can also be thought as giving some evidence to a conjecture of Barak, Kindler and Steurer~\cite{BarakKS13} that no polynomial-time algorithm (including in particular the SOS algorithm) can beat
the basic semidefinite program on approximating random CSP instances.
\end{remark}

Throughout this paper we restrict ourselves to the \emph{Boolean} case, and do not consider extensions to a larger alphabet, though our methods may be useful in this case as well.

\subsection{Related works}

Grigoriev~\cite{Grigoriev01} proved in 1999 that (in the language of this paper)  3XOR is approximation resistant for the degree $\Omega(n)$ Sum-of-Squares hierarchy. Grigoriev's work in fact predated the papers of Parrilo~\cite{Parrilo00}  and Lasserre~\cite{Lasserre01} proposing the SOS hierarchy, and so he used the different (but equivalent) language of Positivstellensatz Calculus proofs. (Also, as far we know, he did not note that these proofs can be efficiently found via a semidefinite program.) Grigoriev's result was rediscovered in 2008 by Schoenebeck~\cite{Schoenebeck08}, who also noted that it implies approximation resistance for 3SAT and some other CSPs  as well.  Tulsiani~\cite{Tulsiani09} (see also Chan~\cite{Chan13}) further generalized these results and in particular showed that every predicate that contains a pairwise independent subgroup is approximation resistant for $\Omega(n)$-degree SOS. Both Tulsiani and Schoenebeck follow Grigoriev's technique of reducing SOS lower bounds to resolution width lower bounds. As far as we know, no other SOS integrality gaps for approximating CSPs were known, and there are very few SOS lower bounds in general, most notably Grigoriev's lower bound for knapsack~\cite{Grigoriev01b} and the very recent result by Meka, Potechin and Wigderson for the planted clique problem (personal communication).

Arora, Bollob{\'{a}}s, Lov{\'{a}}sz and Tourlakis~\cite{AroraBLT06} obtained integrality gaps for the  Lov{\'{a}}sz-Schrijver linear programming hierarchy for Vertex Cover. Schoenebeck, Trevisan  and Tulsiani~\cite{SchoenebeckTT07} showed that Max-Cut is approximation resistant for $\Omega(n)$ levels of the Lov{\'{a}}sz-Schrijver  hierarchy, and these results have been strengthened to the stronger Sherali-Adams hierarchy~\cite{de2007linear,Charikar09}. The famous  Goemans-Williamson algorithm~\cite{GoemansW95} shows that Max-Cut is \emph{not} approximation resistant for even the degree $2$ SOS hierarchy, further underscoring the difference between these relaxations.

Perhaps closest to our work are the papers of Benabbas, Georgiou, Magen, and Tulsiani~\cite{BGMT12} who showed that predicates containing a pairwise independent distribution are approximation resistant for $\Omega(n)$ rounds of the Sherali Adams hierarchy, even when one adds the degree $2$ SOS constraints. Indeed, our pseudo-distribution agrees with theirs, though we describe it somewhat differently, and most importantly, need a completely different argument to show that it is positive semi-definite. Our work is also inspired by the pseudo-expectation view of the SOS hierarchy as advocated in the papers \cite{BarakBHKSZ12,BarakKS14}.
%
%
%
%
%

%% file: overview.tex
\section{Overview of our proof}

To prove Theorem~\ref{thm:main}, we need to show that given any pairwise independent distribution $\mu$ over $\pmo^k$, one can come up with $\I$, a collection of tuples $\{ C_1,\ldots, C_m \}$ of literals and a pseudo-expectation
operator $\pE$ that ``pretends'' to be the expectation of a valid distribution whose projection on to any $C_i$ is $\mu$.
In fact, our choices for both $\I$ and $\pE$ will not be novel and follow prior works in this area. For $\I$, as mentioned, we will simply use a random set of tuples (or more accurately, a set corresponding to a hypergraph with  sufficiently strong expansion properties), as was done by previous works dealing with weaker hierarchies~\cite{BGMT12,TulsianiW12,ODonnellW14}.
It turns out that given this choice, the pseudo-expectation $\pE$ is essentially ``forced'', and again, we use the same pseudo-expectation used in prior works such as~\cite{BGMT12}, though we describe it slightly differently.
This pseudo-expectation corresponds in some sense to the ``maximum entropy distribution'' conditioned on satisfying our constraints (though of course it is not an actual distribution but only a \emph{pseudo-distribution}
in the sense of~\cite{BarakS14}). Those prior works have shown that for every set $S$ of $o(n)$ variables, there is a distribution $\nu_S$ over the variables in $S$ that agrees with $\pE$ .
The main difference is that we prove that for some $d= \Omega(n)$, $\pE$ is a valid degree-$d$ pseudo-expectation operator, that is, it satisfies the non-negativity / positive semidefinite-ness condition $\pE [p^2] \geq 0$ for every polynomial $p\leq d/2$. This is a more ``global'' property, as the polynomial $p$ might depend on all $n$ variables, which makes it more challenging to prove.

Our approach is to essentially diagonalize $\pE$. That is, we will show an explicit construction of polynomials $\tchi_1,\ldots,\tchi_M \in \cP_{d/2}^n$ which we call \emph{local orthogonal functions} such that \textbf{(1)} $\{ \tchi_i \}_{i=1}^M$ spans the space $\cP_{d/2}^n$, \textbf{(2)} $\pE [\tchi_i\tchi_j] = 0$ for all $i\neq j$ and \textbf{(3)} $\pE [\tchi^2_i] \geq 0$ for all $i$.
The existence of these polynomials immediately implies the property we need, as,  by representing every polynomial $p$ as $p = \sum_i p_i \tchi_i$, we see that
\[
\pE [p^2] = \sum_{i,j}  p_ip_j \pE [\tchi_i \tchi_j] = \sum_i p_i^2 \pE [\tchi_i^2] \geq 0 \;.
\]

We now review the construction of the instance, as well as the pseudo-expectation operator, and then discuss how we come up with these local orthogonal functions.
As mentioned above,  our instance $\I = (C_1,\ldots,C_m)$ will simply be a random instance, which we think of as a $k$-uniform hypergraph with $m$ hyperedges $C_1,\ldots,C_m$. After some pruning we can assume this
hypergraph has girth $\Omega(\log n)$.\footnote{If we don't prune these clauses then our proof guarantees that for $1-o(1)$ fraction of the
clauses we get the marginal distribution to be $\mu$. It is possible that this can be upgraded to all of the clauses at the expense of some additional complication, but we have not checked whether or not that's the case.}
By a simple Chernoff + union bound argument, if $m>cn$ for a sufficiently large constant $c$ then for every assignment $x\in\pmo^n$, the induced distribution $\{ C_i(x) \}_{i\sim [m]}$ will be $\e$-close to the uniform distribution. For this informal overview, suppose that we merely want to establish the existence of a degree $d$ pseudo-expectation operator for some large constant $d$. Note that this means that sets of at most $d$ (or even $2^d$) variables form a \emph{forest} (i.e. disjoint collection of trees) in this hypergraph.

We now describe the pseudo-expectation operator $\pE$, which in some sense is almost ``forced'' as the only natural operator for this instance. (As mentioned, this part is not novel and the same operator was used by works such as~\cite{BGMT12}; however we describe it somewhat differently.)
We construct $\pE$ by defining for every set $S$ of at most $d$ variables a distribution $\nu_S$ over $\pmo^S$ such that \textbf{(1)} for every clause $C$ contained in $S$, the projection of $\nu_S$ to $C$ equals $\mu$ and \textbf{(2)} the distributions are \emph{locally consistent} in the sense that if $S\subseteq U$ then the projection of $\nu_U$ to $S$ equals $\nu_S$.
The definition of $\nu_S$ is very simple. First, say for the purposes of this informal overview that a set $S$ is \emph{closed} if every clause $C$ in $\I$ is either completely contained in $S$ or intersects it in at most a single variable. If $S$ of size $O(d)$ is closed and connected (as a subgraph of $\I$) then it is a \emph{tree} in the hypergraph $\I$. In this case, we define the distribution $\nu_S$ as follows: to sample $x$ from $\nu_S$ we pick an arbitrary clause $C \subseteq S$ and sample its variables according to $\mu$. We then continue down the tree, sampling the variables of all the clauses that intersect with $C$, and so on. It is not hard to show that because of pairwise independence (and in fact simply because every marginal is uniform) this process will always yield the same distribution regardless of the traversal order, and the probability of $x \in \pmo^S$ to be sampled under this distribution will be proportional to $\prod_{C\subseteq S}\Pr[ \mu = C(x)]$. If a set $S$ is closed but not connected then the distribution $\nu_S$ is obtained by making independent choices for each of the connected components of $S$. For a general (not necessarily closed) set $S$, we define the \emph{closure} of $S$, denoted by $\cl(S)$, to be the \emph{minimal} closed superset of $S$ (this is well defined; one can show that intersections of closed sets are closed and thus, \emph{the} minimal closed set is the intersection of all closed sets containing $S$).  A fairly simple argument using the girth condition can be used to argue that $|\cl(S)| \leq O(|S|)$ for every $|S|\leq d$. We then define $\nu_S$ to be the distribution obtained by projecting the distribution $\nu_{\cl(S)}$ to $S$. The collection of local distributions so obtained satisfies \textbf{(1)} by construction, and it is not hard to show that it satisfies \textbf{(2)} as well.
Since all polynomials of degree at most $d$ are spanned by the set of polynomials $\{ \chi_S \}_{|S|\leq d}$ (which we will call the \emph{characters}) where $\chi_S(x) = \prod_{i\in S} x_i$, to define the pseudo-expectation operator it suffices to define $\pE [\chi_S]$ for every $|S|\leq d$. We simply define $\pE [\chi_S]$ to be $\E_{x\sim \nu_S}[ \chi_S(x)]$.

\begin{figure}
\begin{center}
\includegraphics[width=3in]{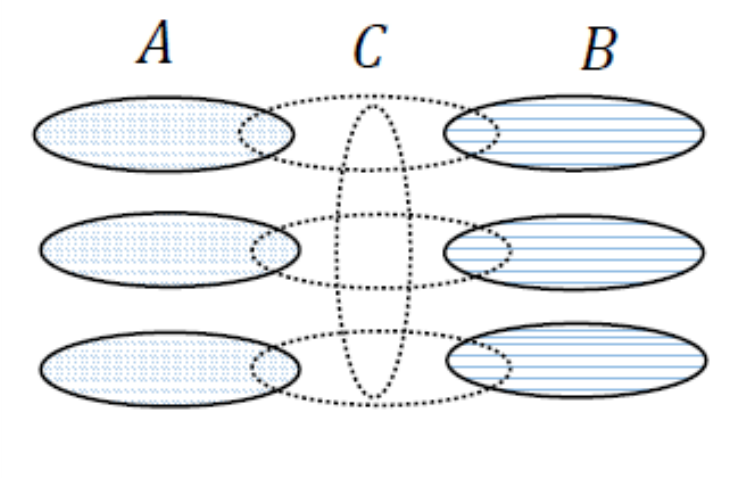}
\end{center}
\caption{In this example, even though both $A$ and $B$ are collections of disjoint clauses and hence are ``closed'' under our definition, their distributions could be correlated due to the existence of the set $C$.}
\label{fig:correlations}
\end{figure}

We now describe how we come up with the functions $\tchi_1,\ldots,\tchi_M$.  Intuitively, we would like to come up with these functions via a Gram-Schmidt like process. That is, we fix some ordering $A_1 \prec \ldots \prec A_M$ of the $M=\binom{n}{\leq d}$ sets of size at most $d$, and define $\chi_i$ to be $\chi_{A_i}$.  Now, we would want to define $\tchi_i$ to be the component orthogonal to the span of $\chi_1,\ldots,\chi_{i-1}$ where we define orthogonality using $\pE$ as an inner product. We would then get that $\pE \tchi_i \chi_j =0$ for all $j<i$, which would imply that $\pE \tchi_i \tchi_j = 0$ for all $i\neq j$ (as $\tchi_j$ is spanned by $\chi_1,\ldots,\chi_j$).
However, this is of course circular reasoning, since we cannot assume that $\pE$ is positive semidefinite (and hence a valid inner product) since this is exactly what we are trying to prove!

However, because we know that on every small set $U$, $\pE$ agrees with an actual expectation operator (the one associated with the \emph{actual} distribution $\nu_{U}$), we do know that it is psd when it is restricted to this small set $U$.
Therefore, if for some reason when we do this Gram-Schmidt process and express $\tchi_i$ as some linear combination $\sum_{j\leq i} \alpha_j \chi_j$, we get lucky and this linear combination happens to be extremely \emph{sparse} then we can actually carry through the argument described above. Specifically, it turns out that it suffices for the set $U=\cup \{ A_j \mid \alpha_j \neq 0 \}$ to be sufficiently small so that $\pE$ is a valid inner product on $U \cup A_i$. However a priori, this hope seems dubious, since the Gram-Schmidt process is very sequential, and we need to do it for $\binom{n}{\leq d}$ steps.
It seems quite possible that we would create \emph{long distance correlations} in the process, whereby we would end up needing to express $\tchi_i$ using many $\chi_j$'s for sets $A_j$ that are quite far from $A_i$. (See Figure~\ref{fig:correlations} for one example of a correlation that could arise between two disjoint collection of clauses $A$ and $B$.)

Nevertheless, we show that we are in fact able to choose a tailor-made ordering of the sets so that this hope is (essentially) materialized.
An important observation that comes to our aid here is that our local distributions, intuitively speaking, satisfy: if two sets $A$ and $B$ are sufficiently far apart in the hypergraph $\I$, then the distribution $\nu_{A\cup B}$ is obtained by taking the product of the independent distributions $\nu_A$ and $\nu_B$. We use this observation to argue that, if we choose the ordering on the sets in ${[n] \choose d}$ in the right way, then, when we express $\tchi_i$ as a linear combination of the functions $\chi_j$ for $j<i$, we
only use $j$'s such that $A_j$ is contained in a certain (carefully defined) small ``ball'' in the hypergraph around the set $A_i$. The crucial result that we need here is to show that whenever there is a dependence between the local distribution on some set $A$ and the local distribution on some set $B$ that came \emph{before} $A$ in our order, then, either $B$ is contained in this ``ball'' around $A$, or the correlation between $A$ and $B$ is completely ``explained" by the intersection  of the closure of $B$ with this ball, in the sense that conditioned on any assignment to the variables in the intersection, the local distributions on $A$ and $B$ are independent. This will allow us to argue that we don't need to use $\chi_B$ to express $\chi_{A_i}$ but can restrict ourselves to characters contained in that ball. Moreover, and crucially,  we will show that our ordering has the property that all the characters we will need to use must have come before $A$ as well.

\begin{figure}
\begin{center}
\includegraphics[width=3in]{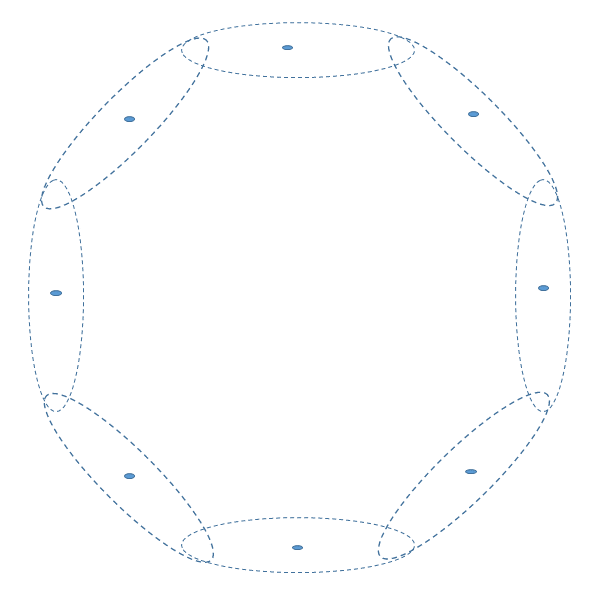}
\end{center}
\caption{In this example, the solid dots are variables and no clause contains any two of them, but the local distribution on the variables might not be uniform since the constraints of the cycle can create a dependency.}
\label{fig:bad-closure}
\end{figure}

\paragraph{Handling $\Omega(n)$ rounds.} 
The above overview can be converted into a full proof with some care when $d = o(\log{(n)})$ by exploiting the acyclicity of all subgraphs involved. Extending to $d = \Omega(n)$ , however, introduces additional subtleties. When $d$ exceeds $\Omega(\log{(n)})$, subgraphs induced by $d$ vertices of $\I$ can have cycles. An immediate effect of this is that the the definition of a closed set that we gave before no longer yields consistent local distributions on any collection of $d$ variables. An example of a problem that arises when cycles can exist on a set of vertices is illustrated in Figure \ref{fig:bad-closure}. To fix this, we define a stronger notion of closed set $S$ that guarantees that all paths of length at most $3$ between any two vertices in $S$ are completely contained inside $S$. This notion of closures differs from the one that Benabbas et.~al.~\cite{BGMT12} use. An appeal to the expansion property of $\I$ (instead of high girth as before) can be used to show that the closure of a set $S$ is at most a constant factor larger than  $|S|$. Similarly, as before, we need to show that there exists a (suitably defined) ball, $Ball(A)$ around any set $A$ of variables (of size at most $d$) such that the correlations with any other set $B$ of size at most $d$ are ``captured'' by the intersection of $Ball(A)$ and $B$. This needs a more careful argument. In particular, the correlations (even in the low girth case) are actually not necessarily captured by the intersection of $Ball(A)$ with $B$, but rather with some set $B'$ that is related, but not identical to $B$. However, the crucial property that we require is that the set $B_{in} = Ball(A)\cap B'$ satisfies \textbf{(1)} if $B$ came before $A$ in the ordering, then so will $B_{in}$ and \textbf{(2)} $|B_{in}| + |B\setminus Ball(A)| \leq |B|$. This second property is more complicated to prove in the case where $|B|$ can be much larger than the girth bound, but turns out to hold there as well.  The bottom line is that with additional care however, the high level picture provided by this overview can indeed be implemented and we give a full analysis based on the local Gram-Schmidt like process in Section \ref{sec:PSD}.

%% file: prelims.tex
\section{Preliminaries} \label{sec:prelims}

We collect some standard definitions and notation here.  A \emph{$(k,n)$-instance} is a $k$-uniform hypergraph $\I = \{ C_1,\ldots, C_m \}$ over $[n]$ so that every hyperedge (also known as a \emph{clause}) $C = (i_1,\ldots,i_k) \in \I$ is labeled by a string $\sigma = \sigma^C \in \pmo^k$. We identify a clause $C$ with the function that maps $x\in \pmo^n$ to $y_1,\ldots,y_k$ where $y_j = \sigma_j x_{i_j}$.
We will sometimes also consider $C$ as a tuple of the \emph{literals}  $(\sigma_{i_1} x_{i_1},\ldots,\sigma_{i_k} x_{i_k})$. We write $V(C)$ for the variables involved in (or \emph{covered by}) a clause $C$ and similarly for $V\subseteq [n]$ we write $\C(V)$ for the set of all clauses $C$ such that $V(C) \subseteq V$. For any $x \in \on^n$, we write $x_A$ to denote the tuple of coordinates in the subset $A \subseteq [n]$. If $x \in \on^A$ and $y \in \on^B$ for disjoint sets $A$ and $B$, we will write $x \circ y$ for the string in $\on^{A \cup B}$ that projects to $x$ for coordinates in $A$ and to $y$ for coordinates in $B$.

Unless explicitly mentioned, the base of all logarithms appearing in the paper is assumed to be $2$. We consider the arity of our tuples $k$ to be a constant and so $O$ notation may hide the dependence on $k$.

We now define some standard ideas in the context of hypergraphs.

\begin{definition}
Let $G$ be a hypergraph. $G$ is said to be a \emph{path} if its hyperedges can be ordered into a sequence $C_1, C_2, \ldots, C_{\ell}$ such that for each $2
\leq i \leq \ell$, $C_i \cap C_{i-1} \neq \emptyset$ and $C_i \cap C_j =
\emptyset$ for every $|i-j| > 1$. $G$ is said to be a \emph{cycle} if it has at
least two hyperedges, and there is a cyclic ordering of its hyperedges $C_0,
C_1, \ldots, C_{\ell-1}$, and there are distinct vertices $v_0, \dots, v_{\ell-1}$ with
$v_i\in C_i\cap C_{(i+1)\bmod \ell}$ for all $i$.
$G$ is said to be a \emph{forest} if it does not contain any cycle.
A forest is a \emph{tree} if it is \emph{connected} (i.e.\ for every two distinct
vertices $u$ and $v$, there is a path $C_1, \dots, C_\ell$ such that $u\in C_1$
and $v\in C_\ell$).

The \emph{degree} of $G$ is the maximum number of hyperedges that intersect with any given hyperedge in $G$. The length of the shortest cycle in $G$ is said to be the \emph{girth} of $G$. For any vertices $u,v$ of a hypergraph $G$, we define the \emph{distance}, $\dist(u,v)$ of $u,v$ in $G$ as the minimum number of hyperedges in any path that joins $u$ and $v$ in $G$. For $S,T$, subsets of vertices, we define $\dist(S,T) \eqdef \min_{s \in S, t \in T} \dist(s,t).$
\end{definition}

Next, we define the notion of expansion in a $k$-uniform hypergraph $G$:
\begin{definition}[Coefficient of Expansion]
A $k$-uniform constraint hypergraph $G$ is said to be $(r, \beta)$-expanding if any collection $\C$ of at most $r$ hyperedges of $G$ cover at least $(k-1-\beta) |\C|$ vertices of $G$, i.e.~ $|\{ v \mid \exists C \in \C \text{, } v \in C\}| \geq (k-1-\beta) |\C|$. We call $\beta$, the \emph{coefficient of expansion} of $G$.
\end{definition}

Let $\I$ be a $(k,n)$ instance. We now describe the properties of the $(k,n)$ instances that we need and give a construction for them in Section \ref{sec:nice} of the Appendix by taking a random instance and removing a few clauses. Specifically, we show the existence of \emph{nice} instances, the ones that satisfy the properties described in the lemma below:

\bnote{are these the conditions we need for the new proof?}
\pnote{Yes, I will point it out clearly here.}
\begin{lemma}
Fix $1 > \epsilon, \delta \geq 0$ and $\gamma \geq e^k k^2$. Then, there exists a $k$-uniform constraint hypergraph $G$ with $\gamma  n$ edges such that for $\eta = (1/\gamma^2)^{2/\delta}$, $1/\tau = 4 \log_2 (\gamma k^2)$, 
$G$:
\begin{enumerate}
\item is $(\eta n, \delta)$-expanding,
\item has girth $\girth \geq \tau \log{(n)}$ 
\end{enumerate}
\end{lemma}

We will use this lemma with any given $\epsilon$ (the soundness slack), $\delta = \frac{1}{200}$ and $\gamma = e^k k^2/\epsilon^2$. We will call the instances that satisfy the conditions of the lemma above as \emph{nice}.

For such instances, it is also easy to prove the soundness part
(part (i)) of Theorem  \ref{thm:main} (see Section \ref{sec:soundness} of the Appendix) which we record in the following lemma.

\begin{lemma}\label{thm:nice}
For every $\e>0$ and $k$, if $n$ is sufficiently large then there exists a nice $(k,n)$-instance $\I$ with the property that
for every $x\in\pmo^n$, the distribution $\{ C(x) \}_{C\in \I}$ is $\e$-close in total variation distance to the uniform distribution on $\pmo^k$.
\end{lemma}


%
%

%% file: closure.tex
\section{Closed sets, and the definition of the pseudo-expectation} \label{sec:closure}

Throughout the rest of this paper we fix  $\I = (C_1,\ldots,C_m)$ to be a nice $(k,n)$ instance with coefficient of expansion $\beta$. Thus whenever we mention edges, paths, or clauses, they will always be with respect to the hypergraph $\I$. In this section, we define a linear operator $\tE$ on $P_s^n$, the linear space of multilinear polynomials on $\R^n$ of degree at most $s = \frac{\eta n}{6}$. We will ensure that the $\tE$ so defined will satisfy $\tE[f(C(x))] = \E [f(\mu)]$ for every clause $C \in \I$ and function $f:\pmo^k\rightarrow\R$. In the next section, we will show that the $\tE$ we define here is in fact a pseudo-expectation operator on $P_d^n$ for $d = \frac{\eta n}{10000k}$ and thus obtain our main result. The $\tE$ operator we use was defined in previous works such as Benabbas et. al. \cite{BGMT12} and later also used by Tulsiani and Worah \cite{TulsianiW12} to study weaker LP/SDP hierarchies. Here, we describe a construction of the same operator in a slightly different way so as to help us in the proof of our main result.

To define $\tE$, it is enough to define $\tE[ \chi_S]$ for characters $\chi_S$ for each $S \subseteq [n]$, $|S| \leq s$, as one can then extend $\tE$ linearly to all of $P_s^n$. To do this, we define a probability distribution $\nu_X$ for every $X \subseteq [n]$, such that $|X| \leq s$, and then set $\tE[\chi_S]$ to be the expectation of $\chi_S$ under $\nu_S$.
\subsection{Closures}

We first define the concept of closed sets that is central to our argument.

\begin{definition}[Closure and closed sets] \stoc{\label{def:closure:app}}{\label{def:closure}}  For every $R\geq 1$, a set $A\subseteq [n]$ is \emph{$R$-closed} if for every
$v,v' \in A$, any path of length at most $R$ between $v$ and $v'$ is contained in $A$.
We say that $A$ is \emph{closed} if it is $3$-closed.

We define the \emph{$R$-closure of $A$}, denoted by $\cl_R(A)$, to be the intersection of all sets $B$ such that $A\subseteq B$ and
$B$ is $R$-closed. The \emph{closure of $A$}, denoted by $\cl(A)$, is the $3$-closure of $A$.
\end{definition}

\begin{remark}
Readers familiar with the definition of closure (or advice set) in the work of \cite{BGMT12} or \cite{TulsianiW12} will find the definition of closure above slightly different. The main difference is that our definition allows us to have some nice properties such as uniqueness and that the intersection of two closed sets is closed, which are very helpful for our proof. We stress however that the actual pseudo-expectation is the same as that of those works.
\end{remark}

Next, we give a constructive definition of closure of a set.
\begin{lemma}
Given $S \subseteq [n]$ and any $R < \min\{ \girth/2, \frac{1}{2\beta} \}$, the $R$-\emph{closure} of $S$ can be obtained by the following procedure run on $S$: Set $A := \emptyset$. For every $v, v' \in V(A) \cup S$ such that there is a path of length at most $R$ between $v$ and $v'$ in $\I$ not contained in $A$, add every clause in the path to $A$. Output $V(A) \cup S$. \label{lem:alternate-closure}
\end{lemma}

\begin{proof}
Observe that the procedure terminates as there are only finitely many clauses. Further, the output is closed by virtue of the termination of the procedure. By induction on the time at which a path is added in the procedure, it is easy to show that every closed set containing $S$ must contain the path. Thus, $V(A)$ is a closed set containing $A$ and every clause $C$ such that $V(C) \subseteq V(A)$ satisfies $V(C) \subseteq \cl_R(S)$. The lemma now follows by the minimality of $\cl_R(S)$.
\end{proof}
\bnote{do we take $v,v'$ in $S$ or in $V(A)\cup S$? i think it needs to be the latter for the set to be closed} \pnote{Yes, that's right, I have corrected the statement of the Lemma now.}

Next, we bound the size of $\cl_R(S)$.
\begin{lemma} \label{lem:closure-size}
For any $R < \min\{ \girth/2, \frac{1}{2\beta} \}$ and $S\subseteq [n]$ such that $|S| \leq \frac{\eta n}{10R}$. Then, $|\C(\cl_R(S))| \leq 2R|S|$ and $|\cl_R(S)| \leq 2Rk|S|$.
\end{lemma}
\bnote{didn't check this proof}
\begin{proof}
Consider the procedure described in Lemma \ref{lem:alternate-closure}. Let $S^{iso} \subseteq \cl_R(S)$ be the isolated vertices in $\cl_R(S)$. Observe that one cannot add any isolated vertices in the procedure and thus $S^{iso} \subseteq S$. Define $S' = S \setminus S^{iso}$. Then, $\cl_R(S) = \cl_R(S') \cup S^{iso}$.

If the process terminates before adding a total of $q  = \frac{|S'|}{\frac{1}{R}-\beta}$ clauses, then there's nothing to prove, since $|S'| \leq |S| \leq \frac{\eta n}{10R}$ yields that $q \leq \frac{\eta n}{5}$. Thus, suppose, for the sake of a contradiction, that the procedure adds $> q$ clauses and let $i^{th}$ round of the procedure be the first round where the number of clauses added exceeds $q$.

Let $\C_i$ be the set of clauses added in the procedure till the $i^{th}$ round and let $S'_i$ be the set of variables obtained by taking the union of variables covered by the clauses added and $S'$. Further, suppose that the $i^{th}$ round adds $q_i$ clauses. Then, $|\C_i| \leq q+q_{i} < \eta n$ and thus, $\C_i$ must satisfy the expansion requirement: $|V(\C_i)| \geq (q+q_i)(k-1-\beta)$. On the other hand, any new path of length $j \leq R$ added in a round adds at most $jk - (j-1) -2$ new vertices. Thus, on an average, every one of the at most $j$ new clauses added in any round of the procedure contribute at most: $k-1-1/j \leq k-1-1/R$ new vertices. Thus, $|S'_i| \leq |S'| + (q+q_i) (k-1-1/R)$.

Now, $$(q+q_i)(k-1-\beta) \leq |V(\C_i)| \leq |S'_i| \leq |S'|+ (q+q_i) \cdot (k-1-1/R).$$ This yields that $|S'| \geq (q + q_i) \cdot (1/R - \beta) > |S'|$ using that $q = \frac{|S'|}{\frac{1}{R}-\beta}$. This is a contradiction.

The size claimed in the lemma now follows by observing that $\frac{1}{R} -\beta \geq \frac{1}{2R}$ and that every clause contributes at most $k$ new variables.
\end{proof}

The following lemma summarizes the simple properties of the closures defined here.
\begin{lemma}[Simple Properties of Closures] \stoc{\label{lem:closure-props:app}}{\label{lem:closure-props}}
\begin{enumerate}
\item For any $R < \girth/2$, if $A$ and $B$ are $R$-closed and then so is $A\cap B$.
\item If $A\subseteq B$ then $\cl_R(A)\subseteq \cl_R(B)$.
\item Every connected component of $\cl_R(A)$ of size $\geq 2$ intersects $A$ in at least two elements.
\item Let $A = A_1 \cup A_2 \cup \ldots A_m $. Then, $\cl(A) = \cl( \cup_{i = 1}^m  \cl(A_i))$.
\end{enumerate}
\end{lemma}
\begin{proof}
\begin{enumerate}
\item If there are two vertices $v,v'$ in $A \cap B$ such that $\dist(v,v') \leq R$, then since both $A$ and $B$ are closed, both of them should contain the  unique (since $R<\girth/2$) path between them.



\item By definition, $\cl_R(B)$ is an $R$-closed set containing $B\supseteq A$ and hence if $\cl_R(A)\nsubseteq \cl_R(B)$ then $\cl_R(A)\cap \cl_R(B)$ would be an even smaller $R$-closed set that contains $A$, contradicting the minimality of $\cl_R(A)$.

\item Suppose otherwise that there is some connected component $S$ of $\cl_R(A)$ with $|S|\geq 2$ intersecting $A$ with at most one element $\{ x \}$,
then we claim that $B = (\cl_R(A)\setminus S) \cup \{ x \}$ is an $R$-closed set containing $A$.  Clearly, $B\supseteq A$.
Now suppose for the sake of contradiction that there were two vertices $v \neq v'$ of distance at most $R$ in $B$ whose path is not in $B$.
Then since $B \subseteq \cl_R(A)$ and $\cl_R(A)$ is  $R$-closed, the path between $v$ and $v'$ must have had a vertex $u\in S \setminus \{ x \}$. But since one of $v$ or $v'$ must be different than $x$ (say $v'$), we get by contradiction that $v'$ was connected to $S$ in $\cl_R(A)$.

\item Let $B = \cl ( \cup_{i = 1}^m \cl(A_i))$. Since $\cl(A)$ is closed and contains $\cup_{i = 1}^m A_i$, $B \subseteq \cl(A)$. If $B \neq \cl(A)$, then, $B \supseteq \cup_{i = 1}^m A_i$ and is closed contradicting the minimality of $\cl(A)$.  \qedhere
\end{enumerate}
\end{proof}

\subsection{Definition of $\tE$}
Using the closures defined above, we define a local probability distribution on all closed sets and use it to define $\tE$. Let $C = (v_1, v_2, \ldots, v_k)$, where, each $v_j$ is the literal $\sigma_j x_{i_j}$ for some $\sigma_j\in \pmo$. The distribution $\mu_C$ simply assigns to $x\in \pmo^n$ the probability $\mu(\sigma_1 x_{i_1},\ldots, \sigma_k x_{i_k})$ (i.e., the probability that $C(x)=a$ under $\mu_C$ is set to $\mu(a)$ for every $a\in\pmo^k$).

The definition and the proof of consistency of the local distribution we define were shown by Benabbas et.~al.~\cite{BGMT12} for the weaker notion of closures they used (in order to define linear round solutions in the Sherali Adams hierarchy). The argument for our notion of closure is similar but we include it here for the sake of completeness.

For every set $S \subseteq [n]$, $|S| \leq d$, let $\cl(S)$ be the closure of $S$ and suppose $I_S$ is the set of isolated variables in $\cl(S)$. Define $\C(\cl(S))$ be all clauses $C$ such that $V(C) \subseteq \cl(S)$.
Then, we set:
\begin{equation} \label{eq:def-prob}
\nu_{\cl(S)}(x) = {Z_{\cl(S)}} \cdot \Pi_{C \in \C(\cl(S))}
\mu_C(x_{C})
\end{equation}
where $x_C$ the projection of $x$ on to the coordinates in $V(C)$, and $Z_{\cl(S)} = 2^{k|\C(\cl(S))|-|\cl(S)|}$ ($\geq 1)$.\bnote{this number is larger than one, so better to use notation that makes it clear that's the case}  Observe that the above expression tells us that the marginal distribution of $\nu_{cl(S)}$ over $I_S$ is uniform. We extend the notation above and write $\nu_T$ for the marginal of $\nu_{\cl(T)}$ on variables in $T$.

We now show that $\nu_{\cl(S)}$ defined above is indeed a probability distribution over $\cl(S)$.
\begin{lemma}
Let $A$ and $B$ be closed sets such that $A \subseteq B$ and $|\C(B)| \leq \eta n$. Then,
\begin{enumerate}
\item $\nu_{A}$ is a valid probability distribution: $\sum_{x \in \on^{A}} \nu_{A}(x) = 1$.
\item $\nu$ is locally consistent: for every $x \in \on^S$, $\nu_{A}(x) = \sum_{ y \in \on^{B \setminus A}} \nu_{B}(x \circ y)$.
\end{enumerate}
 \label{lem:prob-normalization}
\end{lemma}

The following claim that we record as a lemma will be useful in the proof.
\begin{lemma} \label{lem:order-consistency}
There exists an ordering $C_1, C_2, \ldots, C_r$ of clauses in $\C_{A,B}$ and a partition of $B \setminus A$ into sets $F_1 \subseteq V(C_1), F_2 \subseteq V(C_2), \ldots, F_r \subseteq V(C_r)$ such that for every $j \leq r$, $|F_j| \geq k-2$ and $F_j \cap \left( \cup_{i > j} V(C_i) \right) = \emptyset$.
\end{lemma}

We first complete the proof of Lemma~\ref{lem:prob-normalization} and then prove Lemma~\ref{lem:order-consistency}.

\begin{proof}[Proof of Lemma~\ref{lem:prob-normalization}]
Let $Z_A = 2^{-|A| +k |\C(A)|}$ and $Z_B = 2^{-|B| + k |\C(B)|}$. Let $\C_{A,B} = \C(B) \setminus \C(A)$. Using \eqref{eq:def-prob}, we have:

\begin{align*}
\sum_{y \in \on^{B \setminus A}} \nu_{B}(x \circ y) &= Z_B \cdot \Pi_{C \in \C(B)} \mu_C(x_C \circ y_C)\\
&= Z_B \cdot \Pi_{C \in \C(A)} \mu_C(x_C) \cdot \Pi_{C \in \C_{A,B}} \mu_C (x_C \circ y_C)\\
\end{align*}

To simplify notation, we will write $\mu_{i}$ for $\mu_{C_i}$ and $x^i$ for $x_{V(C_i)}$ where $x \in \on^n$. We have, using the ordering given by Lemma~\ref{lem:order-consistency}. Then,

\begin{align*}
\sum_{y \in \on^{B \setminus A}} \nu_{B}(x \circ y) &= Z_B \sum_{ y \in \on^{B \setminus A}} \cdot \Pi_{C \in \C(A)} \mu_C(x_C) \cdot \Pi_{C \in \C_{A,B}} \mu_C (x_C \circ y_C)\\
&= Z_B \sum_{y \in \on^{B \setminus A}} \Pi_{C \in \C(A)} \mu_C(x_C) \cdot \Pi_{i = 1}^r \mu_{i}(x^{i} \circ y^{i})\\
\text{Using the} & \text{ partition $F_1, F_2, \ldots, F_r$ }\\
&= Z_B \Pi_{C \in \C(A)} \mu_C(x_C) \cdot \sum_{\alpha_r \in \on^{F_r}} \mu_{r}(\zeta_r \circ \alpha_r)\cdots \sum_{\alpha_1 \in \on^{ F_1}} \mu_{1}(\zeta_1 \circ \alpha_1)\\
\text{ where $\zeta_r$ is the value} &\text{ for the variables in $V(C_r) \setminus F_r$.}\\
\text{ Using that $|F_r| \geq k-2$} &\text{ and pairwise independence of $\mu$}\\
&=  Z_B \Pi_{C \in \C(A)} \mu_C(x_C) \cdot \sum_{\alpha_r \in \on^{F_r}} \mu_{r}(\zeta_r \circ \alpha_r)\cdots \sum_{\alpha_2 \in \on^{ F_2}} \mu_{2}(\zeta_2 \circ \alpha_2) \cdot 2^{-|V(C_r) \setminus F_r|}\\
&\text{ Continuing similarly for $2,3,\ldots,r$}\\
&=  Z_B \Pi_{C \in \C(A)} \mu_C(x_C) \cdot 2^{-\sum_{i = 1}^r |V(C_r) \setminus F_r|}.
\end{align*}

Now, $\sum_{i = 1}^r |V(C_r) \setminus F_r| = kr - |B\setminus A|$. Further, $-|B|+k|\C(B)| -kr+|B \setminus A| = -|A| +k |\C(A)|$. Thus, $Z_B \cdot  2^{-\sum_{i = 1}^r |V(C_r) \setminus F_r|} = Z_A$ completing the proof.  \qedhere

We now complete the proof of Lemma \ref{lem:order-consistency}.
\begin{proof}[Proof of Lemma \ref{lem:order-consistency}]
For every $C \in \C_{A,B}$ define $\Gamma(C) = \{ v \in V(C) \mid \forall C' \neq C \in \C_{A,B} \text{, } v \notin V(C') \}$. For any collection $\C$ of clauses in $\C_{A,B}$, let $\Delta(\C) = |\cup_{C \in \C} \Gamma(C)|$. Similarly, define $\Gamma_A(C) = \Gamma(C) \setminus A$ and $\Delta_A(\C) = |\cup_{C \in \C} \Gamma_A(C)|$. We make the following claim:
\begin{claim} \label{clm:glkfdhgkjfd}
For any $\C \subseteq \C_{A,B}$, $\Delta_A(\C) \geq (k-5/2-2\beta)|\C|.$
\end{claim}
We first complete the proof of the lemma using the claim. Since $\Delta_A(\C_{A,B}) \geq (k-5/2-2\beta)|\C_{A,B}|$ and $\beta <1/10$, there exists a clause $C$ such that $|\Gamma_A(C)| \geq k-2$. Now $V(C) \setminus A \supseteq \Gamma_A(C)$ and thus $|V(C) \setminus A| \geq k-2$. We place this clause at the beginning of the ordering, call it $C_1$ and set $F_1 = V(C) \setminus A$. We now iterate with $\C_{A,B} \setminus \{C\}$ to complete the construction, obtain a clause $C_2 \in \C_{A,B} \setminus C_1$ such that $|\Gamma_A(C_2)| \geq k-2$. Since $\Gamma_A(C_1)$ cannot intersect $\Gamma_A(C_2)$, we can now set $F_2 = V(C_2) \setminus V(C_1)$. Continuing this way yields the required ordering and partition of $B \setminus A$. \qedhere

We now complete the proof of the claim.
\begin{proof}[Proof of Claim~\ref{clm:glkfdhgkjfd}]
Fix any $\C$ and consider any (maximally) connected subgraph with edges $\C' \subseteq \C$. If $\C'$ consists of a single clause $C$, then $|V(C) \cap A| \leq 1$ (since $A$ is closed) and $V(C) \cap V(C') = \emptyset$ for any $C' \neq C \in \C$. Thus, $\Gamma_A(\C') \geq k-1$.

Now suppose $\C'$ consists of at least $2$ clauses. We first claim that $\Delta(\C') \geq (k-2-2\beta)|\C'|$. To see this, observe that $\C'$ is a collection of at most $\eta n$ clauses in $\I$ and thus, $|V(\C')| \geq (k-1-\beta)|\C|$. Further, every $v \in V(\C') \setminus \cup_{C \in \C'} \Gamma(C)$ belongs to at least two different clauses in $\C'$ and thus, $(k-1-\beta)|\C'| \leq |V(\C')| \leq \Delta(\C') + (k|\C'| - \Delta(\C'))/2$. Rearranging gives $\Delta(\C') \geq (k-2-2\beta)|\C'|$.

Next, we claim that that for every $v \in V(\C') \cap A$ there exists a pair of clauses $C,C'$ such that $ V\left( C \cup C' \right) \cap A = \{v\}$. Consider any clause $C \in \C$ such that $V(C) \cap A = \{v\}$. If there is another clause $C'$ such that $V(C') \cap A = \{v\}$, then observe that $V(C')$ cannot intersect $A$ in any other element (since $A$ is closed) and thus we can let $C,C'$ be the pair as above, corresponding to $v$. Otherwise, there exists a clause $C'$ such that $C' \in \C$ such that $V(C') \cap V(C) \neq \emptyset$ (since $V(\C')$ is connected) and $V(C') \cap A = \emptyset$ (as otherwise there would be a path between two distinct vertices of $A$, of length at most $2$ outside of $A$). Further, observe that all such pairs are disjoint. This is because if some pairs intersect, then they induce a path of length at most $3$ between two distinct vertices of $A$ that is not contained in $A$ (violating the $3$ closedness of $A$). Thus, $|V(\C') \cap A| \leq |\C'|/2$. Thus, we must have: $\Delta_A(\C') \geq \Delta(\C') -|\C'|/2 \geq (k-2-2\beta)|\C'| -|\C'|/2 = (k-5/2-2\beta)|\C'|$.

Since for every connected component $\C'$ inside $\C$ we have that $\Delta_A(\C') \geq (k-5/2-2\beta)|\C'|,$ we must have $\Delta_A(\C) \geq (k-5/2-2\beta)|\C|$ as promised. This completes the proof of claim.
\end{proof}
\end{proof}

\end{proof}

\subsection{$\tE$ and some basic properties}

The following is immediate from \eqref{eq:def-prob}:
\begin{lemma} \label{lem:disjoint-closure-indep}
Suppose $A$ and $B$ are closed disjoint sets such that $A \cup B$ is closed. Then, $\nu_{A \cup B} (x) = \nu_{A} (x_A) \cdot \nu_{B}(x_B)$.
\end{lemma}

We now define the pseudo-expectation operator associated with the local distributions $\{ \nu_T \}_{|T| \leq s}$:
\begin{definition}[Pseudo-Expectation]
For the collection of consistent local probability distributions $\{\nu_T \}_{|T| \leq s}$ defined in \eqref{eq:def-prob} for $s \leq \eta n/6$, we define $\tE$ on $P_s^n$ by $$\tE[ \chi_S] = \E_{\nu_S}[\chi_S],$$ for every $|S| \leq s$.
\end{definition}

\begin{corollary} \label{cor:girth-SA}
Let $\I$ be a nice $(k,n)$ instance and $\mu$ a pairwise independent distribution over $\{\pm 1\}^k$. Then the family of local distributions $\{\nu_X\}_{X\subseteq [n], |X| < d}$ for $s = \eta n/6$ satisfies:
\begin{enumerate}
\item Completeness: For every clause $C$ of $\I$, $\nu_{V(C)} = \mu$.

\item Consistency: for every $S \subseteq T\subseteq [n]$, $|T| \leq d$, the
marginal of $\nu_T$ on to $S$ is $\nu_S$.
\end{enumerate}
\end{corollary}

\begin{proof}
The completeness property follows from \eqref{eq:def-prob} and  $\C(V(C)) = \{C\}$.
The consistency property follows from Lemma \ref{lem:prob-normalization}.
\end{proof}

Finally, since $\tE$ corresponds to a valid expectation locally, we obtain that $\tE$ induces a positive semidefinite (PSD) inner product on any space of functions of a small number of variables.
\begin{lemma}[Local PSDness] \stoc{\label{lem:local-PSD:app}}{\label{lem:local-PSD}}
Let $\tE$ be the pseudo-expectation operator defined by the local distributions $\{ \nu_S \}_{|S| \leq s}$. Let $T$ be a subset of $[n]$ of size at most $s$. Then for every $f\in V = \Span \{ \chi_A \mid A\subseteq T\}$, $\tE[f^2] \geq 0$.
\end{lemma}

%% file: independence.tex
\section{Local Distribution on Unions}

In this section we make an important step towards showing the positivity property of our pseudo-distribution by showing that if two sets $A$ and $B$ are sufficiently closed, then the local distribution on $A\cup B$ is only determined by the clauses that are contained in $A$ or in $B$. In particular, this implies that if $A$ and $B$ are disjoint then the distribution on $A$ is \emph{independent} of the distribution of $B$.  The main result of this section is the following expression for the local distribution on the union of $A$ and $B$  where $A$ is $R$-closed for a sufficiently large constant $R$ and $B$ is closed.
\begin{figure}
\begin{center}
\includegraphics[width=3in]{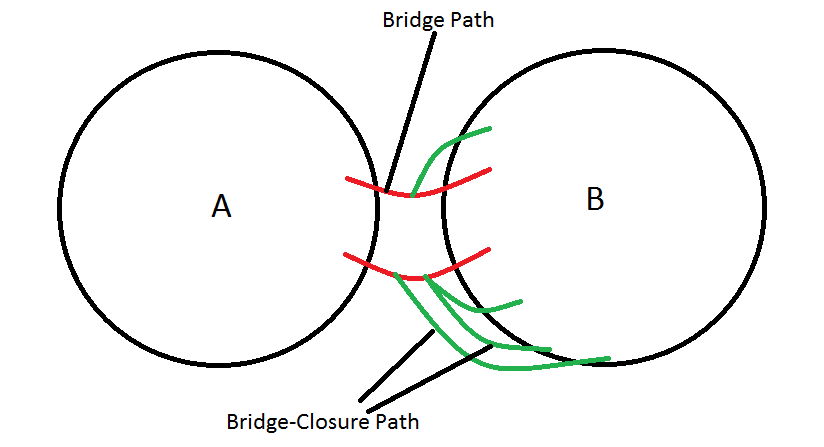}
\end{center}
\caption{A possible configuration when $A$ is $R$-closed and $B$ is closed. All solid lines indicate paths of length at most 3. }
\label{fig:disc-closure-1}
\end{figure}



\begin{lemma}[Local Distribution on Union of Two Closures]
Suppose $A$ is $R$-closed for $R \geq  100$ and $B$ is closed. Then, for any $x \in \on^{A \cup B}$, $$\nu_{A \cup B}(x) = Z_{A,B} \cdot \Pi_{C:V(C) \subseteq A \cup B} \mu_C(x_C),$$ where $Z_{A,B} = 2^{k |\C(A \cup B)| - |A \cup B|}$. \label{lem:local-dist-union}
\end{lemma}

\bnote{do we know of a counterexample when $A$ is merely $3$-closed? If so, perhaps we can show a figure with it to give some intuition}
\pnote{Added a figure for a counter example.}
\bnote{I actually means a counterexample showing two disjoint 3-closed sets $A$ and $B$ such that the distributions on $A$ and $B$ are not independent. (Rather than simply a counterexample showing you can have other clauses than bridges and bridge-closures.) Do we have such a thing?}
\pnote{I tried a bit to construct a counter example, but haven't been able to till now.}

We make two convenient definitions before proceeding, see Figure~\ref{fig:disc-closure-1}:
\begin{definition}[Bridge Paths]
For any two closed sets $A$ and $B$, a path $P$ of length at most $3$ is said to be a bridge path for the pair $A, B$ if $|P \cap A| = |P \cap B|  = 1$.
\end{definition}

\begin{definition}[Bridge-Closure Paths]
For any two closed sets $A$ and $B$, a path $P$ of length at most $3$ is said to be a bridge-closure path for the pair $A,B$, if there exists a bridge path $P'$ such that $|P' \cap P | = 1$ and $|P \cap B| = 1$ but $C \cap A = \emptyset$.
\end{definition}

\paragraph{Proof overview.} Since the proof is rather technical, let us start with a high level overview of it. 
We first show the only extra clauses added to $\cl(A\cup B)$ come from bridge and bridge-closure paths. Moreover, all these additional paths are disjoint apart from their end points. What this amounts to is that the new connections between $A$ and $B$ can be thought of as a collection of disjoint \emph{trees} $T_1,\ldots,T_r$ such that each of these trees has a root in $A$ and its leaves in $B$. The marginal distribution  over $A\cup B$ is obtained by summing up all possible assignments to the intermediate nodes in these trees. Thus at the heart of the proof is the observation that for every such tree $T$ with root $x_0$ and leaves $x_1,\ldots,x_{\ell}$, if we consider the distribution over the variables of $T$ induced by the tree (i.e., where the probability of $x$ is proportional to $\prod_{C\in\C(T)} \mu_C(x_C)$) then the marginal distribution over $\{ x_0,x_1,\ldots,x_{\ell} \}$ is uniform. Hence these trees create no dependence between $A$ and $B$. 

As a final remark, observe that the example from Figure \ref{fig:correlations} shows that $A$ and $B$ being $2$-closed is not enough to guarantee the statement of the lemma. While we believe that at least one of the sets out of $A$ and $B$ should be $R$-closed for some $R > 3$ for the lemma to hold, currently, we do not have any example of a counter example demonstrating this point. We now proceed with the actual proof.

\begin{proof}[Proof of Lemma~\ref{lem:local-dist-union}]
Let $D = \cl(A \cup B)$. Let $\C_{A,B}$ and $\C_{B}$ be the set of bridge paths and bridge closure  paths of $B$ for the pair $A,B$, respectively. Observe that $V(\C_{A,B}) \cup V(\C_{B}) \subseteq D$. We now show that these are the only \emph{extra} clauses in $D$:

We first make a few simple observations:

The first observation describes how bridge paths and bridge-closure paths intersect.

\begin{claim}[Intersections]
\begin{enumerate}
\item For any distinct $P_1, P_2 \in \C_{A,B}$, $P_1 \cap P_2 \subseteq A \cup B$.
\item For any distinct $P_1, P_2 \in \C_B$, $P_1 \cap P_2 \subseteq V(P) \cup B$ where $P$ is a bridge path.
\item For any $P \in \C_B$ and $P' \in \C_{A,B}$, $|V(P) \cap V(P')| \leq 1$.
\item Suppose $P_1, P_2 \in \C_B$ are such that $P_1 \cap P \neq \emptyset$ and $P_2 \cap P' \neq \emptyset$ for some bridge paths $P \neq P'$. Then, $P_1 \cap P_2 = \emptyset$.
\end{enumerate}
\label{claim:bridge-dont-intersect}
\end{claim}
\begin{proof}
\begin{enumerate}

\item If the claim weren't true, then there must be a path of length $\leq 6$ between two vertices of $A$ which violating that $A$ is $R$-closed.
\item Suppose first that there is a bridge path $P$ such that $P \cap P_1 \neq \emptyset$ and $P \cap P_2 \neq \emptyset$. If either of $P_1$ or $P_2$ intersect $P$ in more than one element, then there is a cycle of length at most $6$ in $G$ which violates the fact that $G$ has $\Omega(1)$ girth. If $P_1$ and $P_2$ intersect in an element not contained in $V(P)$, then, again there is a cycle of length at most $9$ in $G$ violating the high girth of $G$. Similarly, if $P_1, P_2$ intersect inside $B$, then, they cannot intersect outside of $B$ and further, they cannot both intersect the same bridge path as it would yield a cycle of length at most $9$ in $G$. Thus in both the cases, $P_1 \cap P_2 \subseteq V(P) \cup B$ for some bridge path $P$.

\item Otherwise there is a cycle of length at most $6$ in $G$ violating that $G$ has girth $\omega(1)$.

\item If not, then if $|P\cap P' \cap A|=1$ then there is a cycle of length $12$ in the graph, contradicting our assumption on the girth. Otherwise $|P \cap P' \cap A|=2$ which means that there is a path of length at most $12$ between two distinct vertices of $A$.
\end{enumerate}
\end{proof}

%

The next observation shows that there is no path of length at most $3$ that connects two bridge paths, two bridge-closure paths or two bridge-bridge-closure paths that are not contained in $A \cup B$.

\begin{claim}[No Extra Paths]
\begin{enumerate}
\item There is no path $P$ of length at most $3$ not contained in $A$ that connects a bridge path $P'$ and $A$.
\item There is no path of length at most $3$ not contained in $A$ that connects $P \in \C_{A,B}$ with $P' \in \C_B$.
\item There is no path of length at most $3$ connecting distinct $P,P' \in \C_B$.
\end{enumerate}
\end{claim}
\begin{proof}
\begin{enumerate}
\item Otherwise there is a path of length at most $6$ between two vertices of $A$ not contained in $A$, violating the fact that $A$ is $R$ closed.
\item Otherwise there is a path of length at most $12$ between two vertices of $A$, violating that $A$ is $R$ closed.
\item Otherwise there is a path of length at most $18$ not contained in $A$, connecting two vertices of $A$.
\end{enumerate}
\end{proof}

%
%
%
%
%
%
%
%
%
%

The following claim is now a consequence of the claims above:

\begin{claim}\label{claim:bridge-clauses}
For any $C$ such that $V(C) \nsubseteq A \cup B$ but $V(C) \subseteq D$, $C \in \C_{A,B} \cup \C_B$.
\end{claim}

\begin{proof}[Proof of Claim]
Consider the iterative procedure of building the closure of $A \cup B$ by adding paths one by one in some order. Let $P$ be the first path in this order that violates the claim. Then, either $P$ intersects two bridge paths or a bridge path and $A$ or a bridge path and a bridge-closure path or two bridge-closure paths. Each of these situations is explicitly barred by the claims above. This completes the argument.
\end{proof}


Let $Z = 2^{k|\C(D)|-|D|} = 2^{k|\C(A \cup B)| + k |\C_{A,B}| - |D|}$. Observe that $Z \cdot 2^{-2|\C_{A,B}|} = Z_{A,B}$. For every clause $C \in \C_{A,B} \cup \C_B$, let $V_C' = V(C) \setminus (A \cup B)$ and $V_C'' = V(C) \cap (A \cap B)$. Similarly, let $D' = D \setminus (A \cup B)$ and $D'' = D \cap (A \cup B)$. Next, we claim:

\begin{claim}
$$Z \cdot \sum_{\gamma \in \on^{D'}} \Pi_{C \in \C_{A,B} \cup \C_B} \mu_{C}\left(x_{V_C''} \circ \gamma_{V_C'} \right) = Z_{A,B}.$$
\label{claim:big-calculation}
\end{claim}

\begin{proof}

Let $D' = V_1 \cup V_2$ such that $V_1 \cap V_2 = \emptyset$ defined by $V_1 = D' \setminus V(\C_{A,B}))$ and $V_2 = D' \setminus V_1$.

\begin{align*}
Z \cdot &\sum_{\gamma \in \on^{D'}} \Pi_{C \in \C_{A,B} \cup \C_B} \mu_{C}\left(x_{V_C''} \circ \gamma_{V_C'} \right) \\
&= Z\sum_{\gamma \in \on^{D'}} \Pi_{C \in \C_B} \mu_{C}\left(x_{V_C''} \circ \gamma_{V_C'} \right) \Pi_{C \in \C_{A,B}} \mu_{C}\left(x_{V_C''} \circ \gamma_{V_C'} \right) \\
&= Z\sum_{\gamma \in \on^{V_2}}  \Pi_{C \in \C_{A,B}} \mu_{C}\left(x_{V_C''} \circ \gamma_{V_C'} \right) \sum_{\gamma \in \on^{V_1}} \Pi_{C \in \C_B} \mu_{C}\left(x_{V_C''} \circ \gamma_{V_C'} \right)\\
&\text{ Now, observe that for every $C \in \C_{B}$, $V(C) \cap V_2$ has at most $2$ elements. Thus, by pairwise independence of $\mu$ }\\
&= Z\sum_{\gamma \in \on^{V_2}}  \Pi_{C \in \C_{A,B}} \mu_{C}\left(x_{V_C''} \circ \gamma_{V_C'} \right) \Pi_{C \in \C_B} 2^{-|V(C) \cap (A \cup B \cup V_1)|}\\
&\text{ Similarly, for every $C \in \C_{A,B}$, $V(C) \cap (A \cup B)$ contains at most $2$ elements. Thus, }\\
&=Z \Pi_{C \in \C_{A,B}} 2^{-|V(C) \cap (A \cup B)|} \Pi_{C \in \C_B} 2^{-|V(C) \cap (A \cup B \cup V_1)|}\\
&= Z_{A,B}.
\end{align*}

\end{proof}
We can now write, using \eqref{eq:def-prob}:

\begin{align*}
\nu_{A \cup B}(x) &= Z \cdot \sum_{\gamma \in \on^{D'}} \Pi_{C: V(C) \subseteq D} \mu_{C}\left( x_{V_C''} \circ \gamma_{V_C'} \right)\\
\text{ Using} & \text{ Claim \ref{claim:bridge-clauses}}\\
&= Z \sum_{\gamma \in \on^{D'}} \Pi_{C: V(C) \subseteq \left(A \cup B \right)} \mu_{C}( x_{C}) \cdot \Pi_{C \in \C_{A,B} \cup \C_B} \mu_{C}\left(x_{V_C''} \circ \gamma_{V_C'} \right)\\
&\text{ Using Claim \ref{claim:big-calculation} }\\
&= Z_{A,B} \Pi_{C: V(C) \subseteq \left(A \cup B \right)} \mu_{C}( x_{C}).
\end{align*}

This completes the proof.
%
\end{proof}

%
%
%
%

%% file: proof.tex
%
%
%

\section{$\tE$ is positive semidefinite}
\label{sec:PSD}
In this section, we prove our main result. Our proof will follow easily from the following lemma which is the main result of this section.

\stoc{
\begin{lemma}[Main Lemma] \label{lem:orthonormalization}
Let $P_{d}^n = \Span \{ \chi_A \mid |A| \leq d\}$ be the space of multilinear polynomials on $\R^n$ of degree at most $d = \frac{\sqrt{\girth}}{10k^2}$. There exists a collection of functions  $\{ \tchi_i \mid 0 \leq i \leq M\} \subseteq P_{d}^n$ for $M = {n \choose {\leq d}}-1$ such that \textbf{(i)} $P_d^n = \Span \{ \tchi_i \mid 0 \leq i \leq M\}$, \textbf{(ii)} $\tE[ \tchi_i^2] \geq 0$, and \textbf{(iii)} $\tE[ \tchi_i \cdot \tchi_j] = 0$ whenever $i \neq j$.
\end{lemma}
}
{
\begin{lemma}[Main Lemma] \label{lem:orthonormalization}
Let $P_{d}^n = \Span \{ \chi_A \mid |A| \leq d\}$ be the space of multilinear polynomials on $\R^n$ of degree at most $d = \frac{\eta n}{10000k}$. There exists a collection of functions  $\{ \tchi_i \mid 0 \leq i \leq M\} \subseteq P_{d}^n$ for $M = {n \choose {\leq d}}-1$ such that:
\begin{enumerate}
\item $P_d^n = \Span \{ \tchi_i \mid 0 \leq i \leq M\}$.
\item $\tE[ \tchi_i^2] \geq 0$.
\item $\tE[ \tchi_i \cdot \tchi_j] = 0$ whenever $i \neq j$.
\end{enumerate}
\end{lemma}
}
We first complete the proof of of Theorem~\ref{thm:main} assuming this lemma. Observe that part (1) of the theorem follows from Lemma \ref{lem:soundness}. Further, $\tE$ satisfies $\tE[ f(C_i)] = f(\mu)$ by Corollary \ref{cor:girth-SA}. Thus, we only need to prove that $\tE$ is a valid pseudo-expectation operator, that is, that $\tE$ is positive semidefinite.

Let $f \in P_d^n$ be any multilinear polynomial of degree $\leq d$. Then, we show that $\tE[ f^2] \geq 0$. We use the spanning property (1) of the $\tchi_i$s above to write $f = \sum_{i <{n \choose {\leq d}}} f_i \cdot \tchi_{i}$. Using orthogonality (3) of $\tchi_i$s, we have: $\tE[ f^2] = \sum_{i \in {[n] \choose {\leq d}}} f_i^2 \tE[ \tchi_{i}^2]$. Finally, using the positivity property (2) of the $\tchi_i$s, we have that $\tE$ is PSD.

The rest of this section is devoted to proving Lemma \ref{lem:orthonormalization}.

\subsection{Choosing an Ordering}

Our aim is to build an order on the ${[n] \choose {\leq d}}$, in which to process them for our local orthogonalization procedure.  We start with an arbitrary ordering on the clauses of $\I$, e.g. for every $C\in \I$ we define a unique index $\zeta(C) \in [m]$. We say that $A \prec B$ if:
\begin{itemize}

\item $\C(\cl(A))$ is smaller than $\C(\cl(B))$ in lexicographic order of $\zeta$. That is, $A \prec B$ if the maximum
index $\zeta(C)$ for $C\in \cl(A)$ is smaller than this maximum for $\cl(B)$, and if they are equal we break ties by the second largest index and so on. We define $\pi(\cl(A))$ to be the index of $\cl(A)$ according to this ordering. (Note that $\pi$ is a permutation on \emph{distinct} closures, and so if $\cl(A) \neq \cl(B)$ then $\pi(\cl(A)) \neq \pi(\cl(B))$.)

\item If $\C(\cl(A))=\C(\cl(B))$ then we say that $A \prec B$ if $|A|<|B|$.

\item If $\C(\cl(A)) = \C(\cl(B))$ and $|A|=|B|$ then we break ties arbitrarily.

\end{itemize}

For $i=0,\ldots, M$, we let $A_i$ denote the $i^{th}$ set in this ordering. Note that $A_0 = \emptyset$ and
$A_1,\ldots,A_n$ are the singleton elements $\{1 \},\ldots,\{ n \}$ (in some arbitrary order).
We will write $\chi_i$ for $\chi_{A_i}$ in the following to reduce clutter.

\subsection{Local Orthogonalization}

Set $R = 100$. Define the \emph{$i^{th}$ local correlated space} as \stoceq{V_i = \Span \{ \chi_B \mid |B| \leq d, B \subseteq \cl_R(A_i), B \prec A_i \}.}

\begin{lemma}
For every $f\in V_i$,  $\tE[ f^2] \geq 0$. \label{lem:PSDVi}
\end{lemma}

\begin{proof}
Invoking Lemma \ref{lem:local-PSD}, it suffices to show that $|\cl_R(A_i)| < s = \eta n/6$. This follows by noting that $|A_i| \leq d$ and $|\cl_R(A_i)| \leq 2Rkd = 200 \eta n/10000 = \eta n/500 \leq s$ (Lemma \ref{lem:closure-size}).
\end{proof}

Define $\bchi_i$ to be any $f \in V_i$ such that $\tE[ (\chi_{i}-f)^2] \leq \tE[(\chi_{i}-g)^2]$ for every $g \in V_i$. Note that such a function must exist because $\tE[ (\chi_i - f)^2 ] \geq 0$ for every $f$ (one can WLOG minimize on the orthogonal complement of the kernel of $\tE$ inside $V_i$).  We define $\tchi_i = \chi_i - \bchi_i$. Since $V_0$ is empty, we set $\bchi_0$ as the constant $0$ function and $\tchi_0$ is thus defined as $\chi_0 = \chi_\emptyset=1$.

The following simple lemma would be very useful.

\begin{lemma}[Local orthogonality]
$\pE [\tchi_i g] = 0$ for every $g \in V_i$.
\label{lem:local-projections-exist}
\end{lemma}

\begin{proof}
Since both $g$ and $\bchi_i$ are spanned by characters of size at most $d$ and $2d < s$, the pseudo-expectation is well defined. Further, since both $g$ and $\bchi_i$ lie in $\Span\{ \chi_S \mid S \subseteq \cl_R(A_i) \}$ and $|\cl_R(A_i)| \leq s$  (as in the proof of Lemma \ref{lem:PSDVi}), $\tE$ corresponds to the expectation operator associated with the probability distribution $\nu_{\cl_R(A_i)}$.

Now suppose for the sake of contradiction that \stoceq{
\tE[ (\tchi_i - \bchi_i)g ] = \delta
}
for some $\delta > 0$. (If the expectation is negative then we can take $-g$.)
Let $f = \bchi_i - \e g$. We have:
\[
\tE[ (\chi_i - f)^2 ] = \tE[ (\chi_i - \bchi_i)^2 ] + \e^2 \tE[ g^2 ] -2\e \delta
\]
and so if $\e$ is sufficiently small then \stoceq{
\tE[ (\chi_i - f)^2 ] < \tE[ (\chi_i - \bchi_i)^2 ]
}
contradicting our choice of $\bchi_i$.
\end{proof}

%

The following lemma shows that the $\tchi_i$'s span $P_d^n$:

\begin{lemma}
For every $i$: $\Span \{ \tchi_j: j \leq i \} = \Span \{ \chi_{j}: j \leq i \}$. \label{lem:equal-spans}
\end{lemma}
\begin{proof}
First, we show that for every $i$, $\chi_{i} \in \Span \{ \tchi_j \mid j \leq i \}$. We argue by induction. $\tchi_1 = \chi_1$ and thus the statement holds for $i = 1$. Now suppose the statement holds for all $j < i$.  Consider $\chi_i$. From the definition of $\tchi_i$ above, we have that: $\chi_i  = \tchi_i+\bchi_i$. Now, $\bchi_i \in V_i$ and $V_i \subseteq \Span \{ \chi_j \mid j <i \}$ by definition. Further, by inductive hypothesis, each $\chi_j$ for $j < i$ satisfies $\chi_j \in \Span \{ \tchi_{j'} \mid {j'} \leq j \} \subseteq \Span \{ \tchi_{j'} \mid {j'} < i\}$. This completes the induction.

The other direction is easier: $\tchi_i = \chi_i-\bchi_i$ and as argued above, $\bchi_i \in \Span \{ \chi_j \mid j <  i \}$. Thus, $\tchi_i \in \Span \{ \chi_j \mid j \leq i\}$.  Together, we thus have: $\Span \{ \tchi_j \mid j \leq i\} = \Span \{ \chi_j \mid j \leq i \}$.
\end{proof}

\subsection{Global Orthogonality lemma}

In this section, we prove the following lemma that is the technical heart of the proof and says that local orthogonalization is enough to ensure that $\tchi_{i}$ are all mutually orthogonal.

\begin{figure}
\begin{center}
\includegraphics[width=3in]{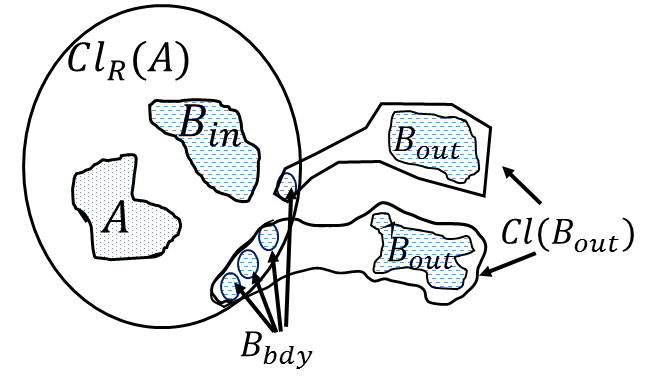}
\end{center}
\caption{A possible configuration of $B_{in}$, $B_{out}$ and $B_{bdy}$.}
\label{fig:boundary}
\end{figure}

\begin{lemma} \label{lem:crucial}
For every $j < i$, \stoceq{\tE[ \tchi_{i} \cdot \chi_{j}] = 0.}
\end{lemma}

We will need the following observation for the proof which we record before proceeding:

\begin{lemma} \label{lem:long-dist-many-edges}
Suppose $H$ is a connected $k$-uniform hypergraph such that there exist a subset of vertices, $U$, $|U| \geq 2$ satisfying: $\dist(u,v) > R$ for every distinct $u,v \in U$. Then, $H$ must have at least $\frac{|U|R}{2}$ hyperedges.
\end{lemma}
\begin{proof}
Observe that the collection of balls of radius $R/2$ around any vertex in $u \in U$ are all disjoint and contain at least one path (due to connectedness of $H$).
\end{proof}

We now go on to prove Lemma \ref{lem:crucial}.

\begin{proof}
Fix any $j <i$ and let $A=A_i$ and $B = A_j$. Let\\
\stoceq{
B_{bdy} = \{ x \in \cl(B) \mid \exists \text{ a clause } v(W) \in \C(\cl(B))  \text{ s.t.  } W \cap \cl_R(A) = \{x\} \}.
}
For every $x \in B_{bdy}$ we call any associated clause $W$ as in the definition above as a \emph{boundary clause}. Let $B_{out} = B \setminus \cl_R(A_i)$ and $B_{in} = B \setminus (B_{out} \cup B_{bdy})$ and $B_{rest} = B \setminus (B_{out} \cup B_{in})$. Note that $B_{bdy}$ is not necessarily a subset of $B$. Next, we make two useful claims:
\begin{claim} \label{claim:size}
\stoceq{\left| B_{bdy} \cup B_{in} \right| \leq |B| \leq d.}
\end{claim}

\begin{proof}
We will show that $|B_{bdy}| \leq |B_{out}|$. This immediately yields the claim by observing that $d \geq |B| = |B_{in}|+ |B_{out}| + |B_{rest}| \geq |B_{in}| + |B_{bdy}|$. We note that the proof of this claim is significantly simpler in the case that $|B| < R/2$. Proving it in the case when $R$ is a constant and $|B|=\Omega(n)$ is one of the main technical ingredients in getting the proof sketched in the overview to work for  $\Omega(n)$ rounds of the SOS hierarchy.

Let $Q \subseteq [n]$ be a (maximally) connected component in the subgraph defined by the hyperedges $\C(\cl(B)) \setminus \C(\cl_R(A))$. Let $Q_{bdy} = B_{bdy} \cap Q$ and $Q_{out} = B_{out} \cap Q$. $B_{bdy}$ is thus partitioned into $Q_{bdy}$ for every possible maximally connected subgraphs $Q$. It is thus enough to prove that $|Q_{bdy}| \leq |Q_{out}|$ for any fixed $Q$.

Observe that $Q \cap \cl_R(A) = Q_{bdy}$. If $Q \cap \cl_R(A) = \emptyset$, then, there is nothing to prove. If $Q_{bdy} = \{ v\}$, then,  $Q$ contains $V(W_v)$ where $W_v$ is a boundary clause associated with $v$. If $Q$ contains no vertex of $B_{out}$, then, observe that $\cl(B) \setminus (Q \setminus \{v\})$ is a closed set containing $B$ contradicting the minimality of $\cl(B)$. Thus, in this case, $|Q_{bdy}| \leq |Q_{out}|$.

Now suppose for $|Q_{bdy}| \geq 2$. Then, vertices in $Q_{bdy}$ are connected through clauses in $Q$. On the other hand, since $A$ is $R$-closed, for any $u,v \in Q_{bdy}$, any path that uses clauses from $Q$ between $u,v$ must be of length at least $R+1$. Applying Lemma \ref{lem:long-dist-many-edges}, we observe that $|\C(Q)| \geq |Q_{bdy}|R/2$.

Next, we claim that $Q \subseteq \cl(Q_{bdy} \cup Q_{out})$. It is easy to complete the proof once we have this claim: observe that $$|Q_{bdy}|R/2 \leq |\C(Q)| \leq |\C(\cl(Q_{bdy} \cup Q_{out}))| \leq 6|Q_{bdy}| + 6|Q_{out}|.$$ Rearranging yields that $|Q_{out}| \geq |Q_{bdy}| \cdot \frac{R/2-6}{6}$. Using $R \geq 24$ yields that $|Q_{out}| \geq |Q_{bdy}|$. 

We now proceed to show that $Q \subseteq \cl(Q_{bdy} \cup Q_{out})$. By Lemma \ref{lem:closure-props} (4), $\cl(B) = \cl(B_{in} \cup B_{bdy} \cup B_{out})$. Let $B' = B_{in} \cup B_{bdy} \cup B_{out} \setminus (Q_{bdy} \cup Q_{out})$. Then, by another application of Lemma \ref{lem:closure-props} (4), $\cl(B) = \cl( \cl(B') \cup \cl(Q_{bdy} \cup Q_{out}))$. In other words, one can build the closure of $B$ by first building the closure of $B'$ and $Q_{bdy} \cup Q_{out}$ (Step $1$) and then taking the closure of the unions of the obtained sets (Step $2$). Clearly, the final output contains every clause in $\C(Q)$. If we show that (1) $\C(\cl(B')) \cap \C(Q) = \emptyset$ and that (2) no clause from $\C(Q)$ is added in the step $2$, then every clause in $\C(Q)$ must be added in the procedure to build $\cl(Q_{bdy} \cup Q_{out})$ and thus we are done. We now proceed to show the two statements above.

(1): First observe that $\cl(B')$ itself can be built by building the closure of $B_{in}$ (and $\cl(B_{in}) \subseteq \cl_R(A) \Rightarrow \C(\cl(B_{in}) )\cap \C(Q)  = \emptyset$), the closure of $B_{out} \cup B_{bdy} \setminus (Q_{bdy} \cup Q_{out})$ (that cannot intersect any clause from $\C(Q)$ as then $Q$ must include a vertex from $B_{out} \cup B_{bdy} \setminus (Q_{bdy} \cup Q_{out})$, a contradiction) and finally taking the closure of their union. This last step cannot add a clause in $Q$: every path $P$ added connects $\cl(B_{in})$ and $\cl(B_{out} \cup B_{bdy} \setminus (Q_{bdy} \cup Q_{out}))$. If $P$ is contained in $\cl_R(A)$, then, there is nothing to prove. Otherwise $P$ must pass (exactly once) through a boundary vertex. If $P$ contains a clause from $\C(Q)$, then, if $P$ passes through a boundary vertex not in $Q_{bdy}$, then this enlarges $Q$ violating that $Q$ is a maximally connected component. If, on the other hand, $P$ passes through a boundary vertex in $Q_{bdy}$, then, $P$ connects $B_{out} \setminus Q$ with $Q$ violating the maximality of $Q$. Thus, $\C(\cl(B'))$ cannot include any clause from $\C(Q)$.

(2): Consider the step $2$ of the procedure to build $\cl(B)$. In this step, we add paths (of length at most $3$) that connect $\cl(B')$ and $\cl(Q_{bdy} \cup Q_{out})$. For any such path $P$, if $P$ includes some clause $C$ from $\C(Q)$ then it crosses out of $\cl_R(A)$ (exactly once) and thus must pass through a boundary vertex. By maximality of $Q$, we must have that $P \cap B_{bdy} \in Q_{bdy}$ and $P\setminus \C(\cl_R(A)) \subseteq \C(Q)$. On the other hand, the part of $P$ that connects some vertex in $Q_{bdy}$ to $\cl(Q_{bdy} \cup Q_{out})$ is of length at most $3$ and thus must be contained in $\cl(Q_{bdy} \cup Q_{out})$. Thus every edge in $P \setminus \C(\cl_R(A))$ is present in $\C(\cl(Q_{bdy} \cup Q_{out})$ and thus $C \in \C(Q)$.

\end{proof}

\begin{claim} \label{claim:ordering}
Suppose $B_{out} \neq \emptyset$. Then, for every $S \subseteq B_{in} \cup B_{bdy}$, $S \prec A$.
\end{claim}
\begin{proof}
Since $B_{out} \neq \emptyset$, $\cl(B) \neq \cl(A)$. Thus, $\pi(\cl(B)) < \pi(\cl(A))$. Now, $|B_{in} \cup B_{bdy}| \leq d$ from Claim \ref{claim:size}. Thus, every subset $S \subseteq B_{in} \cup B_{bdy}$ has a well-defined ordering w.r.t ${[n] \choose {\leq d}}$. Further, for every such $S$, $\cl(S) \subseteq \cl(B)$ (Lemma \ref{lem:closure-props}) and thus, $\pi(\cl(S)) \leq \pi(\cl(B))$. Hence, $S \prec A$.
\end{proof}

We now proceed to complete the proof of the lemma. It is easy to verify that $|\cl_R(A) \cup \cl(B)| < s$ and thus by Corollary ~\ref{cor:girth-SA} the $\tE$ operator on functions on variables in $\cl_R(A) \cup \cl(B)$ corresponds to the expectation of a valid local distribution. In what follows, whenever we write $\Pr$, we mean the probability of an event w.r.t. this local probability distribution. (Note that the expectation w.r.t. this probability distribution agrees with $\tE$ whenever both are defined.)

Now,
$\chi_B = \chi_{B_{in}}\chi_{B_{rest}}\chi_{B_{out}}$ and we can write
\begin{equation}
\tE[ \tchi_i \chi_j ] = \pE [ \tchi_i \chi_{B_{in}}\chi_{B_{rest}}\chi_{B_{out}} ] \label{eq:gfdgfdh}
\end{equation}

Consider an arbitrary assignment $z$ to  $B \setminus A $ and  $\gamma \in \pmo^{|B_{bdy}|}$ to $x_{B_{bdy}}$. Let $\1_{B_{bdy}=\gamma}$ be the function that on input $x\in \pmo^n$ outputs $1$ if $x_{B_{bdy}}=\gamma$ and zero otherwise.

Lemma \ref{lem:local-dist-union} gives the expression for the local distribution on $\cl_R(A) \cup \cl(B)$. Using the expression, we have: $$\pE [ \tchi_i \chi_{B_{in}}\chi_{B_{rest}}\chi_{B_{out}} \mid x_{B \setminus A} = z\text{, } x_{B_{bdy}} = \gamma] = \pE [ \tchi_i \chi_{B_{in}}\chi_{B_{rest}} \mid  x_{B_{bdy}} = \gamma] \cdot \chi_{B_{out}}(z_{B_{out}}),$$ where the $\tE$ on the RHS matches the expectation operator associated with the probability distribution $\nu_{\cl_R(A)}$.

We will show that $\pE [ \tchi_i \chi_{B_{in}}\chi_{B_{rest}}\mid x_{B_{bdy}} = \gamma] = 0$ for every choice of $\gamma$. First, we observe that: \begin{equation} \Pr[ \1_{B_{bdy} = \gamma}] \cdot \tE[ \chi_S \cdot \chi_{B_{in}}\chi_{B_{rest}} \mid \1_{B_{bdy} = \gamma}] =  \tE[ \chi_S \cdot \chi_{B_{in}}\chi_{B_{rest}} \cdot \1_{B_{bdy} = \gamma}] , \label{eq:star3} \end{equation} for ever $S \subseteq \cl_R(A)$, $|S| \leq d$.

Now, $\tchi_i \in \Span \{ \chi_S \mid S \subseteq \cl_R(A_i) \text{, } |S| \leq d\}$, and thus using \eqref{eq:star3}, $$\Pr[ \1_{B_{bdy} = \gamma}] \cdot \pE [ \tchi_i \chi_{B_{in}}\chi_{B_{rest}} \mid x_{B_{bdy}} = \gamma] = \tE[ \tchi_i \cdot \chi_{B_{in}}\chi_{B_{rest}} \cdot \1_{B_{bdy} = \gamma}].$$

Now $|B_{in} \cup B_{bdy}| \leq |B_{out}|$ (Claim \ref{claim:size}) and $\1_{B_{bdy}} = \gamma \in \Span \{ \chi_S \mid S \subseteq B_{bdy} \}$: $$\chi_{B_{in}}\chi_{B_{rest}} \cdot \1_{B_{bdy} = \gamma} \in \Span \{ \chi_{B_{in}} \cdot \chi_{B_{rest}} \cdot \chi_T  \mid T \subseteq B_{bdy}\}.$$ Each index set $S$ of the characters above is a subset of $B$  and thus $S \prec A_i$ (invoking Claim \ref{claim:ordering} along with the fact $\pi(\cl(B)) <\pi( \cl(A))$ ). Thus, $\chi_{B_{in}}\chi_{B_{rest}} \cdot \1_{B_{bdy} = \gamma}  \in V_i$. Using Lemma \ref{lem:local-projections-exist}, thus, $$\tE[ \tchi_i \cdot \chi_{B_{in}}\chi_{B_{rest}} \cdot \1_{B_{bdy} = \gamma}] = 0.$$

\end{proof}

We can now complete the proof of Lemma \ref{lem:orthonormalization}.
\begin{proof}[Proof of Lemma \ref{lem:orthonormalization}]
We show that the $\tchi_i$ constructed above satisfy all the properties required. By Lemma \ref{lem:equal-spans}, $\Span \{ \tchi_i\mid i \leq M\} = \Span \{ \chi_i \mid i \leq M \} = P_d^n$.
Next, observe that $\tchi_i = \chi_i - \bchi_i$. Both $\chi_i$ and $\tchi_i$ lie in $\Span \{ \chi_S \mid S \subseteq \cl_R(A_i) \}$ and by Lemma \ref{lem:local-PSD}, $\tE$ is a psd expectation operator over $V_i$. Thus, $\tE[ \tchi_i^2] \geq 0$ for every $i \leq M$. Finally, we verify that $\tchi_i$ are mutually orthogonal. Fix any $i$. It is then enough to show that $\tE[ \tchi_j \cdot \tchi_i] = 0$ for every $j \neq i$. Since $\Span \{ \tchi_r \mid r \leq j \} = \Span \{ \chi_r \mid r \leq j\}$ (Lemma \ref{lem:equal-spans}), we only need to show that $\tE[ \chi_j \cdot \tchi_i] = 0$ for every $j < i$. Invoking Lemma \ref{lem:crucial} then completes the proof.
\end{proof}

%% file: appendix.tex
\section{Random sparse predicates} \label{sec:sparse-predicate}

Consider a random sparse predicate $P$ on $k$ variables and accepting
$|P^{-1}(1)| = t$ assignments.
If $t = \exp(o(k))$, we now show that $P$ does not support a pairwise
independent subgroup with high probability, as $k$ tends to infinity.
Here the randomness corresponds to choosing $P^{-1}(1)$ to be a $t$-sized
subset of $\{0,1\}^k$ uniformly at random.

\begin{observation} \label{obs:no-aff-plane}
$P^{-1}(1)$ does not contain any affine subspace of dimension $2$ (over
$\F_2$) with probability $\geq 1-t^4/2^k$.
\end{observation}

Under the condition of the observation, $P^{-1}(1)$ does not contain any
pairwise independent subgroup, because any such a subgroup contains an affine
subspace of dimension $2$.

\begin{proof}[Proof of Observation~\ref{obs:no-aff-plane}]
Let $v_1, \dots, v_t\in P^{-1}(1)$ be an enumeration of vectors in $P^{-1}(1)$.
Note that if $P^{-1}(1)$ contains a subspace of dimension 2, then there are
$1\leq a < b < c\leq t$ such that this subspace is exactly the affine span of
$v_a, v_b, v_c$.

For a fixed choice of the triple $(a,b,c)$, conditioning on the event that
$v_a, v_b, v_c$ span an affine subspace of dimension $2$, the remaining vector
from this affine subspace also belongs to $P^{-1}(1)$ with probability at most
$t/2^k$.
Taking a union bound over $(a,b,c)$ (at most $t^3$ such choices), we see that
$P^{-1}(1)$ contains an affine subspace with probability at most $t^4/2^k$.
\end{proof}

\section{Constructing nice instances} \label{sec:nice}

In this section, we show the existence of \emph{nice} instances of constraint hypergraphs and prove Theorem \ref{thm:nice}.

\begin{lemma}
Fix $1 > \epsilon, \delta \geq 0$ and $\gamma \geq e^k k^2$. Then, there exists a $k$-uniform constraint hypergraph $G$ with $\gamma  n$ edges such that for $\eta = (1/\gamma^2)^{2/\delta}$, $\tau = 4 \log_2 (\gamma k^2)$, 
$G$:
\begin{enumerate}
\item is $(\eta n, \delta)$-expanding,
\item has girth $\girth \geq \log{(n)}/\tau$, and 
\end{enumerate}
\end{lemma}

\begin{proof}
We first choose a random graph $G$ by choosing every $k$ uniform hyperedge, independently, with probability $ p = 4\gamma \cdot k!/n^{k-1}$. Our final hypergraph will be obtained by removing hyperedges from $G$. 

We first show that:
\begin{claim}
For $G$ chosen as above, with probability at least $1/3$,
\begin{enumerate} 
\item has between $2\gamma n$ and $6 \gamma n$ edges.
\item has $(\eta n, \delta)$-expansion,
\item has at most $n^{1/4} \log{(n)}$ cycles of length at most $\girth$ and
\end{enumerate}
\end{claim}
We first show that the claim above is enough to complete the proof of the lemma. We define $G'$ to be the hypergraph obtained by removing every cycle of length at most $\girth$.
By the claim above, the total number of hyperedges removed in this process, for a large enough $n$, is at most $\gamma n$. Observe that the last  property in the statement of the theorem is immediately satisfied by $G'$. Further, since $G'$ is obtained only by removing hyperedges from $G$, $G'$ still enjoys $(\eta n, \delta)$-expansion. Thus, $G'$ is a constraint hypergraph that satisfies the requirements of the lemma. Finally, the total number of edges removed is sublinear in $n$ and thus $G'$ has at least $\gamma n$ edges for a large enough $n$.
 
We now move on to complete the proof of the claim above:
\begin{proof}[Proof of Claim]

\begin{enumerate}

\item The expected number of edges in $G$ is given by $p \cdot {n \choose k}  = 4 \gamma n (1 - \frac{k-1}{n})^{k-1} \geq 4 \gamma n (1 - \frac{(k-1)^2}{n})$. By an application of Chernoff bound, the probability that the number of edges does not lie in the interval $[2 \gamma n, 6 \gamma n]$ is at most $2e^{\frac{-\gamma n}{16}}$.

\item Next, consider any collection of $s$ clauses and let us compute the probability that they cover at most $cs$ variables for some $c = k-1-\delta$. This probability, is then upper bounded by $$ { n \choose {cs}} \cdot { {{cs} \choose k} \choose s} p^s. $$ Using that ${{cs} \choose k} \leq (cs)^k/k!$ and the approximation ${x \choose y} \leq \left( \frac{xe}{y} \right)^y$, we can upper bound the above expression by: $$ \left( \frac{ne}{cs} \right)^{cs} \cdot \left( \frac{e \frac{\left(cs \right)^k}{k!}}{s} \right)^{s} \left( \frac{ \gamma \cdot k!}{n^{k-1}}\right)^s.$$ Using that $c = k-1-\delta$ and that $\delta < 1$ now yields an upper bound of $$\left( \frac{s}{n} \right)^{\delta s}  \cdot \left( \gamma e^{k} c^{2} \right)^s.$$ Thus, using that $\gamma > e^k k^2$ and that $s$ satisfies $\frac{s}{n} \leq (1/\gamma^2)^{2/\delta}$ makes the above probability at most $(1/\gamma^2)^{s}$.

\item To see how to ensure that the high girth requirement, we first observe that for any integer $\ell$, the expected number of cycles of length $\ell$ in $G$ is at
most $(dk^2)^\ell$.

We first count the number of ways to choose a cycle of length $\ell$.
Recall that a cycle is given by a cyclic sequence $C_1, \dots, C_\ell$ of
hyperedges.
There are ${n \choose k}$ ways to choose $C_1$, and for $2\leq i < \ell$, at
most $k{n \choose k-1}$ ways to choose the common vertex $C_{i-1}\cap C_i$ and
remaining vertices for $C_i$, and finally at most $k^2 {n \choose k-2}$ to
choose $C_\ell$ that intersects both $C_1$ and $C_{\ell-1}$.
Therefore the expected number of length-$\ell$ cycles is at most
\[ {n \choose k} \cdot \left(k{n \choose k-1}\right)^{\ell-2} \cdot k^2{n
\choose k-2} \cdot \left( \frac{4 \gamma(k!)}{n^{k-1}} \right)^\ell \leq {(4\gamma)}^\ell
k^{2\ell} . \qedhere \]

By an application of Markov's inequality, with probability at least $7/8$ over the draw of hyperedges of $G$, the number of cycles of length at most $\girth = \frac
14 \log_{\gamma k^2} n$ are at most  \[ \sum_{\ell \leq \girth} (4\gamma k^2)^\ell \leq \girth n^{1/4}. \] 
%


By a union bound, now, all the three properties above can be ensured with probability at least $1/3$. 

\end{enumerate}

\end{proof}
\end{proof}

\subsection{Soundness} \label{sec:soundness}

In this section, we show that after fixing the underlying hyperedges $G$ of an
instance, with high probability over the literals on constraints, all
assignments are very close to a random assignment.
Here closeness is measured with respect to the distribution $\{C(x)\}$ as one
chooses a uniformly random constraint among all hyperedges of the hypergraph.

Let $G$ be any hypergraph with $m$ hyperedges.
Let $\I$ be an instance with the same underlying hypergraph as $G$, and with
the literals in all clauses be chosen uniformly at random.
We have the following lemma.

\begin{lemma} \label{lem:soundness}
Suppose $m = \Omega(2^{O(k)}\eps^{-2} n)$.
With high probability over the choice of literals, for any assignment $x\in
\{\pm 1\}^n$, the distribution $\{C(x)\}$ with $C$ chosen uniformly at random
in $\I$ is within $\eps$ statistical distance to the uniform distribution over
$\{\pm 1\}^k$.
\end{lemma}

\begin{proof}
Let $\I = (C_1, \dots, C_m)$ be a fixed collection of literals.
Let $\mu_{\I,x}$ denotes the distribution $\{C_i(x)\}$ when $i$ is drawn
uniformly from $[m]$.
For each local assignment $y\in \{\pm 1\}^k$, the probability $\mu_{\I,x}[y]$
that a random local assignment from $\mu_{\I,x}$ equals $y$ is given by
$\E_{i\in [m]}[\1_{C_i(x) = y}]$.

Now suppose the signs of the literals from $\I$ for every constraint are chosen
uniformly at random, keeping the underlying subhypergraph fixed.
Then $\mu_{\I,x}[y]$ is now a random variable depending on the randomness of
the literals.
For each $i$, the indicator $\1_{C_i(x) = y}$ equals $1$ with probability
$1/2^k$, and equals $0$ with the remaining probability (over the randomness of
the signs of the literals on the $i$-th constraint), and the random variables
$\1_{C_i(x) = y}$ are independent of each other for different $i$.
Therefore $\mu_{\I,x}[y]$ is the average of $m$ independent $\{0,1\}$-indicator
random variables, each being $1$ with probability $1/2^k$.
By Chernoff--Hoeffding bound, we have $|\mu_{\I,x}[y] - 1/2^k| > \eta$ with
probability at most $2\exp(-\eta^2 m/2^{k+1})$.
By a union bound over all assignments $x\in \{\pm 1\}^n$, the maximum deviation
of $\mu_{\I,x}[y]$ from $1/2^k$ (over all $x$) exceeds $\eta$ with probability
at most $2\exp(-\eta^2 m/2^{k+1} + n\log 2)$.
Letting $\eta = \eps/2^k$, we see that
\[ \Pr\left[\max_x \left\{ |\mu_{\I,x}[y] - 1/2^k| \right\} \geq \frac\eps{2^k}
\right] \leq \exp(-\Omega(n)) \]
as long as $m = \Omega(2^{O(k)}\eps^{-2}n)$.

Now the distribution $\{C_i(x)\}$ for a random $i\in [m]$ has statistical
distance at least $\eps$ implies that $|\mu_{\I,x}[y] - 1/2^k| \geq \eps/2^k$
for some $y$.
By a union bound over all $y\in \{\pm 1\}^k$, the distribution $\{C_i(x)\}$ is
close in statistical distance to the uniform distribution on $\{\pm 1\}^k$
except with probability $\exp(O(k)-\Omega(n))$, assuming $m =
\Omega(2^{O(k)}\eps^{-2}n)$.
\end{proof}
%